\documentclass[11pt,onecolumn,draftclsnofoot]{IEEEtran} 
\usepackage{fancyhdr}
\usepackage{epsfig}
\usepackage{threeparttable}
\usepackage{epsf,epsfig}
\usepackage{amsthm}
\usepackage{amsmath}
\usepackage{amssymb}
\usepackage{amsfonts}
\usepackage[noadjust]{cite}
\usepackage{dsfont}
\usepackage{subfigure}
\usepackage{color}
\usepackage{verbatim}
\usepackage[margin=1in]{geometry}

\renewcommand{\baselinestretch}{1.2}
\def\Ip{J_P}
\def\Ipp{I_P}
\def\ddp{\delta_P}
\def\sp{_P}
\def\ss{_S}
\def\fa{\frac{2}{\alpha}}
\begin{document}
\newtheorem{theorem}{Theorem}
\newtheorem{acknowledgement}[theorem]{Acknowledgement}
\newtheorem{axiom}[theorem]{Axiom}
\newtheorem{case}[theorem]{Case}
\newtheorem{claim}[theorem]{Claim}
\newtheorem{conclusion}[theorem]{Conclusion}
\newtheorem{condition}[theorem]{Condition}
\newtheorem{conjecture}[theorem]{Conjecture}
\newtheorem{criterion}[theorem]{Criterion}
\newtheorem{definition}[theorem]{Definition}
\newtheorem{example}[theorem]{Example}
\newtheorem{exercise}[theorem]{Exercise}
\newtheorem{lemma}{Lemma}
\newtheorem{corollary}{Corollary}
\newtheorem{notation}[theorem]{Notation}
\newtheorem{problem}[theorem]{Problem}
\newtheorem{proposition}{Proposition}
\newtheorem{solution}[theorem]{Solution}
\newtheorem{summary}[theorem]{Summary}
\newtheorem{assumption}{Assumption}
\newtheorem{examp}{\bf Example}
\newtheorem{probform}{\bf Problem}
\def\remark{{\noindent \bf Remark:\hspace{0.5em}}}

\def\qed{$\Box$}
\def\QED{\mbox{\phantom{m}}\nolinebreak\hfill$\,\Box$}
\def\proof{\noindent{\emph{Proof:} }}
\def\poof{\noindent{\emph{Sketch of Proof:} }}
\def
\endproof{\hspace*{\fill}~\qed
\par
\endtrivlist\unskip}
\def\endproof{\hspace*{\fill}~\qed\par\endtrivlist\vskip3pt}

\def\E{\mathbb{E}}
\def\eps{\varepsilon}
\def\phi{\varphi}
\def\Lsp{{\boldsymbol L}}
\def\Bsp{{\boldsymbol B}}
\def\lsp{{\boldsymbol\ell}}
\def\Ltsp{{\Lsp^2}}
\def\Lpsp{{\Lsp^p}}
\def\Linsp{{\Lsp^{\infty}}}
\def\LtR{{\Lsp^2(\Rst)}}
\def\ltZ{{\lsp^2(\Zst)}}
\def\ltsp{{\lsp^2}}
\def\ltZt{{\lsp^2(\Zst^{2})}}
\def\ninN{{n{\in}\Nst}}
\def\oh{{\frac{1}{2}}}
\def\grass{{\cal G}}
\def\ord{{\cal O}}
\def\dist{{d_G}}
\def\conj#1{{\overline#1}}
\def\ntoinf{{n \rightarrow \infty }}
\def\toinf{{\rightarrow \infty }}
\def\tozero{{\rightarrow 0 }}
\def\trace{{\operatorname{trace}}}
\def\ord{{\cal O}}
\def\UU{{\cal U}}
\def\rank{{\operatorname{rank}}}
\def\acos{{\operatorname{acos}}}

\def\SINR{\mathsf{SINR}}
\def\SNR{\mathsf{SNR}}
\def\SIR{\mathsf{SIR}}
\def\tSIR{\widetilde{\mathsf{SIR}}}
\def\Ei{\mathsf{Ei}}
\def\l{\left}
\def\r{\right}
\def\({\left(}
\def\){\right)}
\def\lb{\left\{}
\def\rb{\right\}}

\setcounter{page}{1}

\newcommand{\eref}[1]{(\ref{#1})}
\newcommand{\fig}[1]{Fig.\ \ref{#1}}

\def\bydef{:=}
\def\ba{{\mathbf{a}}}
\def\bb{{\mathbf{b}}}
\def\bc{{\mathbf{c}}}
\def\bd{{\mathbf{d}}}
\def\bee{{\mathbf{e}}}
\def\bff{{\mathbf{f}}}
\def\bg{{\mathbf{g}}}
\def\bh{{\mathbf{h}}}
\def\bi{{\mathbf{i}}}
\def\bj{{\mathbf{j}}}
\def\bk{{\mathbf{k}}}
\def\bl{{\mathbf{l}}}
\def\bm{{\mathbf{m}}}
\def\bn{{\mathbf{n}}}
\def\bo{{\mathbf{o}}}
\def\bp{{\mathbf{p}}}
\def\bq{{\mathbf{q}}}
\def\br{{\mathbf{r}}}
\def\bs{{\mathbf{s}}}
\def\bt{{\mathbf{t}}}
\def\bu{{\mathbf{u}}}
\def\bv{{\mathbf{v}}}
\def\bw{{\mathbf{w}}}
\def\bx{{\mathbf{x}}}
\def\by{{\mathbf{y}}}
\def\bz{{\mathbf{z}}}
\def\b0{{\mathbf{0}}}

\def\bA{{\mathbf{A}}}
\def\bB{{\mathbf{B}}}
\def\bC{{\mathbf{C}}}
\def\bD{{\mathbf{D}}}
\def\bE{{\mathbf{E}}}
\def\bF{{\mathbf{F}}}
\def\bG{{\mathbf{G}}}
\def\bH{{\mathbf{H}}}
\def\bI{{\mathbf{I}}}
\def\bJ{{\mathbf{J}}}
\def\bK{{\mathbf{K}}}
\def\bL{{\mathbf{L}}}
\def\bM{{\mathbf{M}}}
\def\bN{{\mathbf{N}}}
\def\bO{{\mathbf{O}}}
\def\bP{{\mathbf{P}}}
\def\bQ{{\mathbf{Q}}}
\def\bR{{\mathbf{R}}}
\def\bS{{\mathbf{S}}}
\def\bT{{\mathbf{T}}}
\def\bU{{\mathbf{U}}}
\def\bV{{\mathbf{V}}}
\def\bW{{\mathbf{W}}}
\def\bX{{\mathbf{X}}}
\def\bY{{\mathbf{Y}}}
\def\bZ{{\mathbf{Z}}}

\def\mA{{\mathbb{A}}}
\def\mB{{\mathbb{B}}}
\def\mC{{\mathbb{C}}}
\def\mD{{\mathbb{D}}}
\def\mE{{\mathbb{E}}}
\def\mF{{\mathbb{F}}}
\def\mG{{\mathbb{G}}}
\def\mH{{\mathbb{H}}}
\def\mI{{\mathbb{I}}}
\def\mJ{{\mathbb{J}}}
\def\mK{{\mathbb{K}}}
\def\mL{{\mathbb{L}}}
\def\mM{{\mathbb{M}}}
\def\mN{{\mathbb{N}}}
\def\mO{{\mathbb{O}}}
\def\mP{{\mathbb{P}}}
\def\mQ{{\mathbb{Q}}}
\def\mR{{\mathbb{R}}}
\def\mS{{\mathbb{S}}}
\def\mT{{\mathbb{T}}}
\def\mU{{\mathbb{U}}}
\def\mV{{\mathbb{V}}}
\def\mW{{\mathbb{W}}}
\def\mX{{\mathbb{X}}}
\def\mY{{\mathbb{Y}}}
\def\mZ{{\mathbb{Z}}}

\def\cA{\mathcal{A}}
\def\cB{\mathcal{B}}
\def\cC{\mathcal{C}}
\def\cD{\mathcal{D}}
\def\cE{\mathcal{E}}
\def\cF{\mathcal{F}}
\def\cG{\mathcal{G}}
\def\cH{\mathcal{H}}
\def\cI{\mathcal{I}}
\def\cJ{\mathcal{J}}
\def\cK{\mathcal{K}}
\def\cL{\mathcal{L}}
\def\cM{\mathcal{M}}
\def\cN{\mathcal{N}}
\def\cO{\mathcal{O}}
\def\cP{\mathcal{P}}
\def\cQ{\mathcal{Q}}
\def\cR{\mathcal{R}}
\def\cS{\mathcal{S}}
\def\cT{\mathcal{T}}
\def\cU{\mathcal{U}}
\def\cV{\mathcal{V}}
\def\cW{\mathcal{W}}
\def\cX{\mathcal{X}}
\def\cY{\mathcal{Y}}
\def\cZ{\mathcal{Z}}
\def\cd{\mathcal{d}}
\def\Mt{M_{t}}
\def\Mr{M_{r}}
\def\O{\Omega_{M_{t}}}
\newcommand{\figref}[1]{{Fig.}~\ref{#1}}
\newcommand{\tabref}[1]{{Table}~\ref{#1}}

\newcommand{\var}{\mathsf{var}}
\newcommand{\fb}{\tx{fb}}
\newcommand{\nf}{\tx{nf}}
\newcommand{\BC}{\tx{(bc)}}
\newcommand{\MAC}{\tx{(mac)}}
\newcommand{\Pout}{P_{\mathsf{out}}}
\newcommand{\tPout}{\tilde{P}_{\mathsf{out}}}
\newcommand{\nnn}{\nn\\}
\newcommand{\FB}{\tx{FB}}
\newcommand{\TX}{\tx{TX}}
\newcommand{\RX}{\tx{RX}}
\renewcommand{\mod}{\tx{mod}}
\newcommand{\m}[1]{\mathbf{#1}}
\newcommand{\td}[1]{\tilde{#1}}
\newcommand{\sbf}[1]{\scriptsize{\textbf{#1}}}
\newcommand{\stxt}[1]{\scriptsize{\textrm{#1}}}
\newcommand{\suml}[2]{\sum\limits_{#1}^{#2}}
\newcommand{\sumlk}{\sum\limits_{k=0}^{K-1}}
\newcommand{\eqhsp}{\hspace{10 pt}}
\newcommand{\tx}[1]{\texttt{#1}}
\newcommand{\Hz}{\ \tx{Hz}}
\newcommand{\sinc}{\tx{sinc}}
\newcommand{\tr}{\mathrm{tr}}
\newcommand{\diag}{\mathrm{diag}}
\newcommand{\MAI}{\tx{MAI}}
\newcommand{\ISI}{\tx{ISI}}
\newcommand{\IBI}{\tx{IBI}}
\newcommand{\CN}{\tx{CN}}
\newcommand{\CP}{\tx{CP}}
\newcommand{\ZP}{\tx{ZP}}
\newcommand{\ZF}{\tx{ZF}}
\newcommand{\SP}{\tx{SP}}
\newcommand{\MMSE}{\tx{MMSE}}
\newcommand{\MINF}{\tx{MINF}}
\newcommand{\RC}{\tx{MP}}
\newcommand{\MBER}{\tx{MBER}}
\newcommand{\MSNR}{\tx{MSNR}}
\newcommand{\MCAP}{\tx{MCAP}}
\newcommand{\vol}{\tx{vol}}
\newcommand{\ah}{\hat{g}}
\newcommand{\tg}{\tilde{g}}
\newcommand{\teta}{\tilde{\eta}}
\newcommand{\heta}{\hat{\eta}}
\newcommand{\uh}{\m{\hat{s}}}
\newcommand{\eh}{\m{\hat{\eta}}}
\newcommand{\hv}{\m{h}}
\newcommand{\hh}{\m{\hat{h}}}
\newcommand{\Po}{P_{\mathrm{out}}}
\newcommand{\Poh}{\hat{P}_{\mathrm{out}}}
\newcommand{\Ph}{\hat{\gamma}}
\newcommand{\mat}[1]{\begin{matrix}#1\end{matrix}}
\newcommand{\ud}{^{\dagger}}
\newcommand{\C}{\mathcal{C}}
\newcommand{\nn}{\nonumber}
\newcommand{\nInf}{U\rightarrow \infty}

\title{\huge \setlength{\baselineskip}{30pt} Spatial Interference Cancellation for Multi-Antenna Mobile Ad Hoc Networks}
\author{Kaibin Huang, Jeffrey G. Andrews,  Dongning Guo, \\Robert W. Heath, Jr., and Randall A. Berry \\ (\small Submission date: July 11, 2008)   \thanks{\setlength{\baselineskip}{15pt} K. Huang is with Yonsei University, S. Korea; J. G. Andrews and R. W. Heath, Jr. are with The University of Texas at Austin; D. Guo and R. A. Berry are with Northwestern University. Email: huangkb@yonsei.ac.kr, \{rheath, jandrews\}@ece.utexas.edu, \{dguo, rberry\}@eecs.northwestern.edu. This work was funded by the DARPA IT-MANET program under the grant W911NF-07-1-0028.}}

\maketitle

\begin{abstract}\setlength{\baselineskip}{15pt} 
Interference between nodes is a critical impairment in mobile ad hoc networks (MANETs). This paper studies the role of multiple antennas in mitigating such interference. Specifically, a network is studied in which  receivers apply zero-forcing beamforming to cancel the strongest interferers. Assuming a network with Poisson distributed transmitters and independent Rayleigh fading channels, the transmission capacity is derived, which gives the maximum number of successful transmissions per unit area. Mathematical tools from stochastic geometry are applied to obtain the asymptotic transmission capacity scaling  and characterize the impact of inaccurate channel state information (CSI). It is shown that, if each node cancels $L$ interferers, 
the transmission capacity decreases as  $\Theta(\epsilon^{\frac{1}{L+1}})$
as the outage probability $\epsilon$ vanishes.
For fixed $\epsilon$, as $L$ grows, the transmission capacity increases as
$\Theta(L^{1-\frac{2}{\alpha}})$ where $\alpha$ is the path-loss exponent. Moreover, CSI inaccuracy is shown to have no effect on the transmission capacity scaling as $\epsilon$ vanishes, provided that  the CSI training sequence has an appropriate length, which we derived. 
Numerical results suggest that canceling  merely one interferer  by each node increases the transmission capacity by an order of magnitude or more, even when the CSI is imperfect.
\end{abstract}

\section{Introduction}\label{Section:Intro}
In a mobile ad hoc network (MANET), the mutual interference between nodes poses a fundamental limit on the throughput of peer-to-peer
communication. This paper studies mitigating the effect of interference by provisioning nodes with multiple
antennas. Specifically, each receiver uses zero-forcing beamforming to cancel the interference from the strongest interferers, whereas each  transmitter simply chooses a random beam. This approach requires only limited local coordination and hence is  suitable for a MANET.

This paper considers a simple network consisting of Poisson distributed transmitters and independent Rayleigh fading channels.  We quantify the gains in network performance 
in terms of the {\em transmission capacity} (TC) as a function of the system parameters, including  the  amount and accuracy  of channel state information (CSI) at each receiver. 
The TC is defined in \cite{WeberAndrews:TransCapWlssAdHocNetwkOutage:2005} as the maximum density of successful transmissions so that a typical receiver satisfies an outage probability constraint for a target
signal-to-interference-and-noise ratio (SINR). 
In other words, TC gives the average throughput per unit area.
The derived TC scaling suggests  that multiple antennas can significantly improve the performance of MANETs, even with inaccurate CSI.

\subsection{Prior Work and Motivation}
For Poisson distributed transmitters, TC was introduced in \cite{WeberAndrews:TransCapWlssAdHocNetwkOutage:2005} for single-antenna MANETs assuming fixed transmission power and an ALOHA-like medium access control (MAC) layer. TC has also been used to study  opportunistic transmissions
\cite{WeberAndrews:TransCapAdHocNetwkDistSch:2006}, distributed scheduling \cite{HasanAndrews:GuardZoneAdHocNet:2007}, coverage \cite{GanHaenggi:RegularInterfCap:2006}, network irregularity \cite{VenHaenggi:ShotNoiseCoverage:2006}, bandwidth partitioning \cite{JindalAndrews:BandwidthPartitioning:2007}, successive interference cancellation \cite{WeberAndrews:TransCapWlssAdHocNetwkSIC:2005}, and multi-antenna transmission \cite{AndrewJeff:CapacityScalingSpatialDiversity:2006} in MANETs \cite{HaenggiAndrews:StochasticGeometryRandomGraphWirelessNetworks}. This paper
differs from \cite{AndrewJeff:CapacityScalingSpatialDiversity:2006}  in that the antennas are employed for interference cancellation instead of  interference averaging through  diversity techniques. After the publication  of preliminary results from the current investigation \cite{Huang:SpatialInterfCancel:ImpCSI:Globecom08, Huang:SpatialInterfCancel:PerCSI:Globecom08}, related work has been reported on the TC of multi-antenna MANETs that use  spatial multiplexing and space-time block coding \cite{LouieMacKay:SpatialMultiplexDiversityAdHocNetwork}, space division multiple access \cite{KounAndrews:TCScalingSDMAAdHocNetworks}, or optimally allocate  the spatial degrees of freedom  for  interference cancellation and link enhancement such as spatial multiplexing  \cite{VazeHeath:TransCapacityMultipleAntennaAdHocNetwork} and  array gain \cite{Jindal:RethinkMIMONetwork:LinearThroughput:2008}.

As discussed in \cite{WeberAndrews:TransCapAdHocNetwkDistSch:2006}, the TC is related to the more widely known  \emph{transport capacity}  introduced in \cite{GuptaKumar:CapWlssNetwk:2000}. Transport capacity is typically studied in terms of the scaling of a network's total throughput-distance product as a function of the network size, while TC gives 
 the number of single-hop transmissions possible in a specific area and in terms of the actual design parameters.  Furthermore, most work on the transport capacity assumes perfect scheduling and zero-outage, while we focus on a random access model with an outage requirement.

Besides spatial interference cancellation, there are several alternative approaches for mitigating interference in MANETs.
For example, the \emph{interference alignment} approach in \cite{CadJafar:InterfAlignment:2007} achieves the optimal number of degrees of freedom in a high signal-to-noise ratio (SNR) setting. This approach appears daunting in practice because it requires jointly designed precoders and perfect CSI of interference channels. In contrast, the current approach only requires each receiver to obtain CSI from nearby interferers and no coordination of transmit precoders. Another method for interference management, used in many practical MAC protocols, is to create an interferer-free area -- a \emph{guard zone} -- around each receiving node through carrier sensing. As shown in \cite{HasanAndrews:GuardZoneAdHocNet:2007}, optimizing the guard-zone size leads to significant TC gain for a single-antenna MANET with respect to pure random access.  The use of interference cancellation can be viewed as creating an effective guard zone without requiring that other nearby transmitters be suppressed.

Several papers have addressed other aspects of multi-antenna MANETs. For example, beamforming or directional antennas  have been integrated with the MAC protocols for MANETs to achieve higher network spatial reuse or energy  efficiency \cite{Zorzi:CrossLayerMACMIMOAdHocNetworks:2006,   Hamdaoui07:MultiHopMIMONetworkThput:2007,  Mundarath:CrossLayerAdaptiveAntAdHocNetworks:2007,   Park:SPACE_MAC_MIMO:2005,  Siam:AdaptMultiAntPowrControlWlssNetworks:2006, Ramanathan:AdHocNetworkDirectAntennas:2005, Deopuar:LinkLayerSeriveceDiffWlssNetworksSmartAnt:2007,  Sundar:MACFrameworkAdHocNetworkSmartAnt:2004,  Singh:SmartAlohaMultiHopWlssNetworks:2005,  Ramanathan:PerformAdHocNetworksBeamform:2001}.
In addition, multi-antenna techniques have been applied to enable efficient routing in MANETs \cite{
 Wu:InterestDissemDirectAntWirelessSensor:2006}.
Directional antennas have been studied for suppressing interference in MANETs by  spatial filtering \cite{Ramanathan:AdHocNetworkDirectAntennas:2005,  Deopuar:LinkLayerSeriveceDiffWlssNetworksSmartAnt:2007,  Sundar:MACFrameworkAdHocNetworkSmartAnt:2004,  Singh:SmartAlohaMultiHopWlssNetworks:2005,  Ramanathan:PerformAdHocNetworksBeamform:2001}. Directional antennas, however, are only suitable for environments with sparse scattering and low angular spread. In contrast, beamforming is applicable for both sparse and rich scattering, and is hence adopted in this paper as well as in \cite{Hamdaoui07:MultiHopMIMONetworkThput:2007,   Mundarath:CrossLayerAdaptiveAntAdHocNetworks:2007,   Park:SPACE_MAC_MIMO:2005}
for
spatial interference cancellation. Most prior work focuses on designing MAC protocols and rely on simulations for 
throughput evaluation \cite{KumarRag:MACAdHocNetSurvey:2006,   Zorzi:CrossLayerMACMIMOAdHocNetworks:2006,  Hamdaoui07:MultiHopMIMONetworkThput:2007,
  Mundarath:CrossLayerAdaptiveAntAdHocNetworks:2007,  Park:SPACE_MAC_MIMO:2005,  Siam:AdaptMultiAntPowrControlWlssNetworks:2006,
Ramanathan:AdHocNetworkDirectAntennas:2005,  Deopuar:LinkLayerSeriveceDiffWlssNetworksSmartAnt:2007,  Sundar:MACFrameworkAdHocNetworkSmartAnt:2004,  Singh:SmartAlohaMultiHopWlssNetworks:2005,  Ramanathan:PerformAdHocNetworksBeamform:2001}. In \cite{Zhang:CapImprovWlssAdHocNetworksDirectAnt:2006,  Yi:CapAdHocNetworkDirectAnt:2007}, the use of directional antennas is shown to increase the linear scaling factor of network transport
capacity, which, however,  is too coarse for quantifying the network throughput.   In view of prior work, there still lacks theoretic
characterization of the relationship between the TC of MANETs and spatial interference cancellation. 

\subsection{Contributions and Organization}
Our main contributions are summarized as follows.
\begin{enumerate}
\item Assuming Poisson distributed transmitters and spatially independently and identically distributed (i.i.d.)~Rayleigh fading channels, bounds on the probability of signal-to-interference ratio (SIR) outage are derived with either perfect or imperfect CSI. These bounds are found to be reasonably tight and lead to bounds on TC.
\item 
Let $L$ denote the number of canceled interferers per receiver. 
Irrespective of whether the CSI is perfect,
as the outage probability $\epsilon$ vanishes,
the asymptotic TC is shown to vanish as $\epsilon^{\frac{1}{L+1}}$,
which decays slower  for larger $L$.
\item For fixed $\epsilon$, as $L\rightarrow \infty$, the TC is shown to increase as $L^{1-\frac{2}{\alpha}}$ where $\alpha$ is the path-loss exponent. Hence spatial interference cancellation  is more effective for larger $\alpha$. If instead of canceling interference, the antennas are used 
to maximize the array gain, the TC scales as 
$L^{\frac{2}{\alpha}}$ as shown in \cite{AndrewJeff:CapacityScalingSpatialDiversity:2006}.\footnote{Since this work was submitted, recent results have demonstrated that linear TC scaling  with the number of antennas per node can be achieved by proportionally  allocating  the spatial degrees of freedom at each receiver to obtain array gain  and to cancel interferers, which can be interpreted as the product of above two sub-linear scalings   \cite{Jindal:RethinkMIMONetwork:LinearThroughput:2008}.} 
\item The required training sequence length for CSI estimation is derived for constraining the increase in outage probability and loss in data rate caused by imperfect interference cancellation.   Finally, the required training sequence length is  obtained for  the optimal TC scaling and shown to be proportional to $\log\frac{1}{\epsilon}$. 
\end{enumerate}

Simulation results show that canceling  a few (two to four) interferers per node is sufficient for harvesting most of the available TC gain. A capacity gain of more than an order of magnitude can be achieved by canceling  only one interferer at  each node, even with imperfect CSI. Moreover, a moderate length of the CSI training sequence is observed to be sufficient.

The remainder of this paper is organized as follows. Section~\ref{Section:Sys} describes the network and wireless channel models.
The SIR outage probability and TC are analyzed for perfect and imperfect CSI in Sections~\ref{Section:Outage:PerfectCSI} and \ref{Section:Outage:ImperCSI}, respectively. Numerical results are presented in Section~\ref{Section:Simualtion}. 

\section{Mathematical Models and Metrics} \label{Section:Sys}

\subsection{Network Model}\label{Section:NetworkModel}
The locations of potential transmitting nodes in a MANET are modeled as a 2-D Poisson point process with density $\lambda_o$ following the common approach in the literature (see e.g., \cite{Baccelli:AlohaProtocolMultihopMANET:2006}). Time is slotted and in each time-slot potential transmitting
nodes follow a random access protocol, namely that they transmit independently with a fixed probability $P_t$.  Let $T$ denote the coordinate of a  transmitting node. Given the random access protocol, the set $\acute{\Phi} = \{T\}$ is also a homogeneous Poisson point process but with smaller density $\lambda = P_t\lambda_o$ \cite{Kingman93:PoissonProc}. Each transmitting node is associated with a receiving node located at a unit  distance. Relaxing the assumption that a transmitter and a receiver  are separated by a fixed distance affects the TC scaling  only by a multiplicative factor (see e.g, \cite{WeberAndrews:TransCapAdHocNetwkDistSch:2006}).

Consider a typical receiving node located at the origin, denoted as $R_0$, and let $T_0$ be the corresponding transmitter. This constraint on $R_0$ and  $T_0$ does not compromise generality since the transmitting node process is translation invariant. Furthermore, according to Slivnyak's theorem \cite[Section~4.4]{StoyanBook:StochasticGeometry:95}, the remaining transmitting nodes, namely $\Phi = \acute{\Phi}\backslash\{T_0\}$, remain a homogeneous Poisson point process with the same density $\lambda$.

The MANET is assumed to be interference limited and thus noise is neglected for simplicity.\footnote{Addressing the effect of noise requires straightforward but tedious modifications of the current analysis.
} Consequently, the reliability of data packets received by $R_0$ is determined by the SIR. Moreover, we assume that each data link has a single stream, and communications between nodes are perfectly synchronized at symbol boundaries. All transmitting nodes are assumed to use unit transmission power.

\subsection{Channel Model}\label{Section:ChannelModel}
Let every node be equipped with $N$ antennas, so that the link between each transmitter and receiver can be modeled as an $N\times N$ multiple-input-multiple-output (MIMO) channel.  We assume narrowband channels with frequency-flat block fading.  Moreover, each
MIMO channel consists of path-loss and small-scale fading components. Specifically, the channel from  $T\in\acute{\Phi}$ to  $R_0$ is $r_T^{-\alpha/2}\bG_T$,
 where $\alpha > 2$  is the path-loss exponent, and  $\bG_T$ is an $N\times N$ matrix of i.i.d. unit circularly symmetric complex Gaussian elements (we subsequently denote the distribution as $\mathcal{CN}(0,1)$). Beamforming is applied  at each transmitter and receiver, where $\bff_0$, $\bff_T$ and $\bv_0$ denote the beamforming vectors at $T_0$,  $T\in\Phi$ and $R_0$, respectively. Then, measured at $R_0$, the received power from $T_0$ is  $W = |\bv_0^\dagger\bG_0\bff_0|^2$,  and the interference power from transmitter $T$ is $I_T = r_T^{-\alpha} |\bv_0^\dagger \bG_T\bff_T|^2$,
 where $\dagger$ represents the Hermitian  transpose matrix operation.

\subsection{Transmission Capacity}\label{Section:TxCap:Pre}
Correct decoding of received data packets requires the SIR to exceed a threshold $\theta$, which is identical for all receivers. In other words, the information rate for each link is equal to  $\log_2(1+\theta)$ assuming Gaussian signaling. Note that, with everything else the same, the outage probability increases with the transmitter density $\lambda$.
The TC under the outage constraint is thus given by 
\begin{equation}\label{Eq:TxCap}
    C(\epsilon) = (1-\epsilon)\log_2(1+\theta)\lambda
\end{equation}
where $\lambda$ is such that the outage probability $\Pout$ is $\epsilon$, i.e.,
\begin{eqnarray}
\Pout &=& \Pr(\SIR < \theta) \\
&=& \epsilon. 
\end{eqnarray} 
The metric $C(\epsilon)$ quantifies the spatial reuse efficiency of a single-hop MANET; the product $C(\epsilon) d$ is related to  the transport  capacity of a multi-hop MANET \cite{WeberAndrews:TransCapAdHocNetwkDistSch:2006}.\footnote{A new metric called \emph{random access transport capacity} has been proposed for multi-hop MANETs in a recent work \cite{Andrews:RandomAccessTranspCap}, which   generalizes TC to an end-to-end scenario under several additional assumptions.}

\subsection{Zero-Forcing Beamforming}\label{Section:ZFBeam}
From the perspective of $R_0$, the interference channel from an interferer $T$ is in effect a channel vector $\bh_T = r_T^{-\alpha/2}\bG_T\bff_T$. For convenience, we refer to  the channel norm $J_T =\|\bh_T\|^2$ as the pre-cancellation interference power for $T$. The receiver $R_0$ equipped with $N$ antennas cancels $L$ interferers where $L\leq N-1$. 
Let $\mathcal{T} \subset \Phi$ comprise  the $L$ strongest interferers to be canceled by $R_0$. Spatial interference cancellation  at $R_0$ is realized by choosing  the beamformer $\bv_0$ to be orthogonal to the beams of the $L$ strongest interferers, i.e., $\bv_0^\dagger\bh_T=0$ for every $T\in\mathcal{T}$. 

Conditioned on  interference cancellation, $R_0$ applies the remaining  $(N-L)$ spatial degrees of freedom to enhance the received signal power by maximum ratio combining. To be specific, $\bv_0$ solves  the following optimization problem
\begin{equation}\label{Eq:MRC}
\begin{aligned}
\text{maximize:} &\quad |\bv_0^\dagger\bh_0|^2 \\
\text{subject to:} &\quad \|\bv_0^\dagger \bh_T\|^2=0, \ \forall \ T\in\mathcal{T} \\
&\quad \|\bv_0\|^2=1. 
\end{aligned}
\end{equation}
Note that $\bv_0$ is uniquely determined by $\{\bh_T\mid T\in\mathcal{T}\}$ and $\bh_0$ with probability $1$.


We let transmitter $T_0$ apply a fixed transmit beamformer $\bff_0$ in lieu of an adaptive one (such as to perform maximum ratio transmission \cite{Lo:MaxRatioTx:99}) to prevent the transmitter and receiver from chasing each other's beam, so as to preserve network stability.  Since the fading channels are isotropic, $\bff_0$ is chosen to be the all-one vector normalized by $1/\sqrt{L}$.
With such beamforming, multiple transmit antennas  contribute no additional array gain.

\subsection{The Effective SIR Model: Perfect CSI}\label{Section:PerfectCSI}

We first characterize the SIR at $R_0$ assuming perfect CSI and hence perfect cancellation of the $L$ strongest interferers.
Since $\bG_0$ is isotropic, $\bG_0\bff_0$ is an i.i.d. $\mathcal{CN}(0,1)$ vector.  
  It follows from \cite[Lemma~1]{Jindal:RethinkMIMONetwork:LinearThroughput:2008} that the random variable $W$ has the following chi-square distribution with $(N-L)$ complex degrees of freedom,  
\begin{equation}\label{Eq:PDF:W}
f_W(w) = \frac{w^{N-L-1}}{\Gamma(N-L)}e^{-w}, \quad w\geq 0
\end{equation}
where $\Gamma$ denotes the gamma function. The factor $(N-L)$ specifies the array gain per link \cite{PaulrajBook}.  
The effective interference  model resulting from perfect interference cancellation is illustrated in Fig.~\ref{Fig:EffModel}(a).  
The SIR at $R_0$ is given as
\begin{equation}\label{Eq:SIR_a}
 \SIR = \frac{W}{\sum_{T\in\Phi\backslash\mathcal{T}}  I_T}.
\end{equation}

\begin{figure}
\centering
\includegraphics[width=14cm]{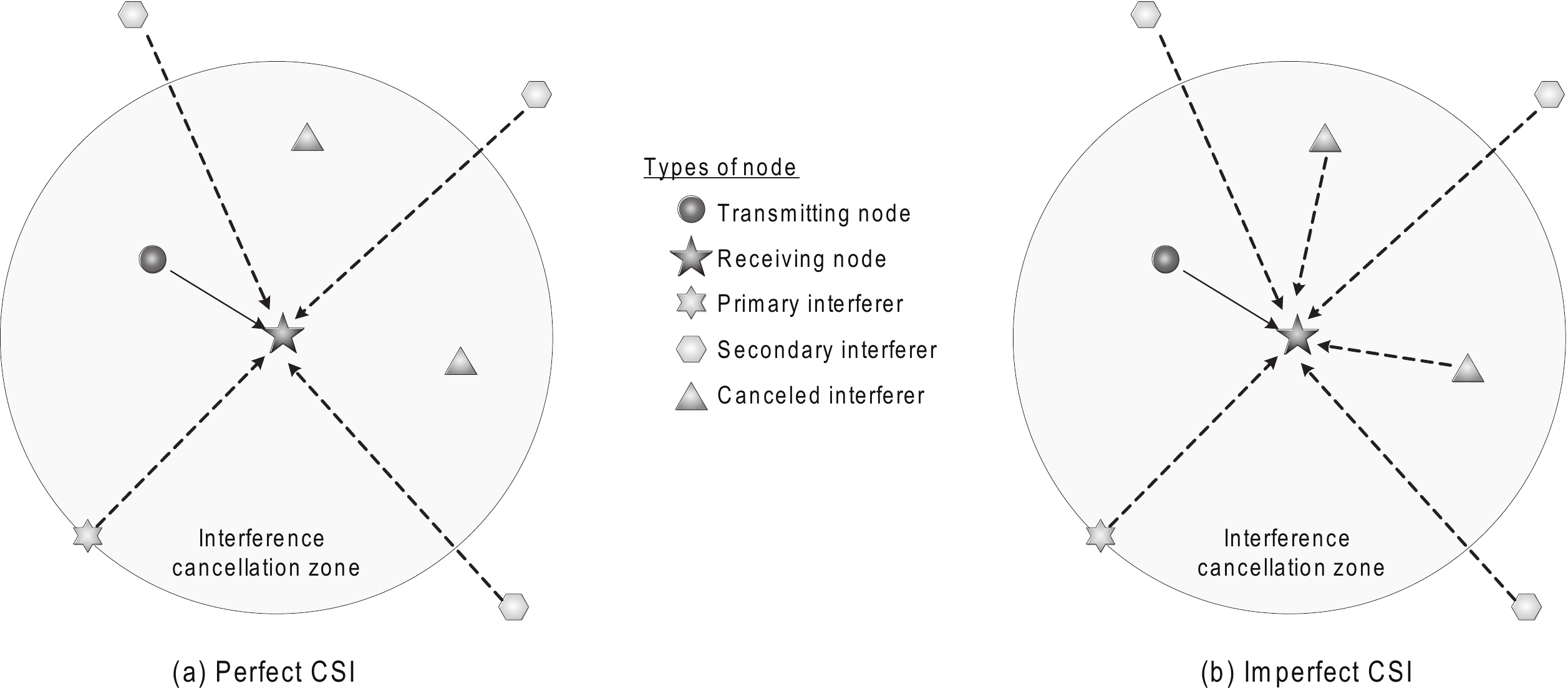}
  \caption{Effective interference  model for a typical receiver canceling two strongest interferers with (a) perfect CSI or (b) imperfect CSI. The distance in the figures decreases with  increasing  received power. The data and interference links are plotted using solid and dashed lines, respectively.}\label{Fig:EffModel}
\end{figure}

\subsection{The Effective SIR Model: Imperfect CSI}\label{Section:ImperfectCSI}\label{Section:IntfAvoid:ImpCSI}

Without the assumption of perfect and readily available CSI,
receivers estimate the signal strength of their interferers and then identify the strongest ones  by using the random training signature sequences inserted into transmitted signals \cite{VerBook}. Subsequently, each receiver requests their strongest interferers to transmit training sequences and uses them to estimate the corresponding interference channels.
Let $M$ denote the length of each training sequence.
The $L$ training sequences form an $L\times M$ matrix $\sqrt{M} \bQ$ 
where $\bQ$ consists of orthonormal row vectors \cite{MarzHoch:FastTransferCSI:2006}. 

Define $\bV$ as an $N\times L$ matrix comprising the vectors $\{\bh_T\mid T\in\mathcal{T}\}$ as columns. The signal $\bY$ received by $R_0$ during the training phase is expressed as an $N\times M$ matrix:
\begin{equation}\label{Eq:Train:Rx}
\bY = \sqrt{M}\bV\bQ +  \sum_{T\in\Phi\backslash\mathcal{T}} \bh_T \bx_T
\end{equation}
where the factor $\sqrt{M}$ normalizes the average power of each training sequence to be one,  and the $1\times M$ row vector  $\bx_T$ contains $\mathcal{CN}(0, 1)$ data symbols transmitted by node $T$. The summation term in \eqref{Eq:Train:Rx} represents interference to the CSI estimation at $R_0$. The CSI is estimated using the least-squares method to yield for every $T\in\mathcal{T}$:
\begin{eqnarray}\label{Eq:CSI:Estim}
\hat{\bh}_T &=& \frac1{\sqrt{M}} \bY\bq_T^\dagger \\
&=& \bh_T + \frac{1}{\sqrt{M}} \sum_{T'\in\Phi\backslash\mathcal{T}} \bh_{T'}\tilde{x}_{T',T} 
\end{eqnarray}
where $\bq_T$ is the training sequence sent by node $T$ and $\tilde{x}_{T',T}=\bx_T\bq_T^\dagger$  
has the distribution $\mathcal{CN}(0,1)$.\footnote{The alternative minimum mean-square-error  (MMSE) estimator requires knowledge of the covariance of the aggregate interference from the weak transmitters, which is difficult to measure accurately due to the presence of the strong interferers.}

The SIR is derived as follows. The estimated CSI is applied for computing the beamformer $\bv_0$ used at $R_0$.  Using \eqref{Eq:CSI:Estim} and under the zero-forcing constraint: $\bv_0^\dagger \hat{\bh}_T=0$, $\forall \ T\in\mathcal{T}$, the residual interference at $R_0$ after beamforming, denoted as $I_{R}$,  can be written as
\begin{eqnarray}
I_{R} &=& \sum_{T\in\mathcal{T}} \bv_0^\dagger\bh_Tx_T \nn\\
&=& - \frac{1}{\sqrt{M}} \sum_{T\in\mathcal{T}} \sum_{T'\in\Phi\backslash\mathcal{T}} 
\bv_0^\dagger \bh_{T'} \tilde{x}_{T',T}x_T.\label{Eq:ResIntfPwr:Pwr}
\end{eqnarray}
 As illustrated in Fig.~\ref{Fig:EffModel}(b), 
CSI estimation errors result in additional interference with respect to the case of perfect CSI. 
For the present case, the SIR in \eqref{Eq:SIR_a} is modified as
\begin{equation}\label{Eq:SIR_b}
\widetilde{\SIR} = \frac{W}{\sigma_{R}^2 + \sum_{T\in\Phi\backslash\mathcal{T}}  I_T}
\end{equation}
where $\sigma_{R}^2$ denotes the variance of $I_{R}$ in \eqref{Eq:ResIntfPwr:Pwr} conditioned on fixed $\Phi$ and channels. 

\section{Transmission Capacity with Perfect CSI}\label{Section:Outage:PerfectCSI}
This section focuses on the analysis of the outage probability and TC assuming prefect CSI. In particular, we derive  the TC scaling   with respect to the number of canceled interferers per node and the outage probability.

\subsection{Point Processes and Auxiliary Results}\label{Section:AuxResult}

\newcommand{\TP}{T_P}

Since the $L$ strongest interferers are canceled, we refer to
the $(L+1)$-st strongest interferer (in terms of pre-cancellation power) as
the \emph{primary interferer} and denote it as $\TP$, whose pre-cancellation
interference power is $J_P = \|\bh_{\TP}\|^2$.
The interferers with smaller pre-cancellation interference power are referred to as the \emph{secondary interferers}. There are two reasons for separating  the interferers. First, considering the primary  interferer alone yields a lower bound on the outage probability to be derived in the sequel. Second, as we show shortly, 
the secondary interferers   conditioned on $\Ip$ form a Poisson point process. 

Recall that $\Phi$ stands for the homogeneous Poisson point process of all transmitters but $T_0$.
Define a marked point process   \cite{Kingman93:PoissonProc} $\Psi = \{(T, J_T\}\mid T\in \Phi\}$ where the mark of node $T$ is its pre-cancellation interference power $J_T$.
Different interference channels are independent and hence $J_T$ depends on $T$ but not other points in $\Phi$.  Let $\mu$ and $\mu^*$ denote  the mean measures of $\Phi$ and $\Psi$, respectively. 
By applying Marking Theorem,   $\Psi$ is shown to be a Poisson process   on the product space $\mathds{R}^2 \times \mathds{R}^+$ with mean measure $\mu^*$ given as \cite{Kingman93:PoissonProc}
\begin{equation}\label{Eq:MMeasure}
\mu^*(\mathcal{B}) = \iint\limits_{(t, u) \in \mathcal{B}}\mu(dt) p(t,du)
\end{equation}
where $\mathcal{B}$ is a measurable subset of $\mathds{R}^2 \times \mathds{R}^+$ and $p(T, \cdot)$ represents the distribution of $J_T$ conditioned on $T$.  
Let $\mathcal{G} = \mathbb{R}^2 \times (g,\infty)$. Note that $\Psi\cap\mathcal{G}$ is the set of interferers whose pre-cancellation interference power is larger than $g$. From \eqref{Eq:MMeasure},    $\mu^*(\mathcal{G})$ can be obtained  as 
\begin{eqnarray}
\mu^*(\mathcal{G})&=& \lambda \int_{t\in\mathds{R}^2}\int_g^\infty p(t, du) dt\label{eq:mua}\\
&=& 2\pi \lambda \int_{0}^{\infty}\int_g^\infty r p(r, du) dr\label{eq:mub}\\
&=& 2\pi \lambda \int_{0}^{\infty}r\Pr(J_T> g \mid |T| = r)dr\label{Eq:MMeasure:a}
\end{eqnarray}
where~\eqref{eq:mua} follows from the homogeneity of $\Phi$, and~\eqref{eq:mub} uses  polar coordinates and the fact that $J_T$  depends only on the distance $r_T = |T|$ from the origin. 

Let the point process consisting of all secondary interferers  conditioned on $\Ip$ be denoted by
\begin{equation}
\Pi(\Ip)  = \{(T, J_T)\,|\, T\in\Phi, 0\leq J_T <  \Ip\}.  
\end{equation}
The distribution of $\Pi(\Ip)$ and $J_P$ are characterized by the following lemmas, which are proved in  Appendices~\ref{App:Second:Dist} and~\ref{App:Primary:PC}, respectively. 
\begin{lemma}\label{Lem:Second:Dist}
Conditioned on  $\Ip=g$,  $\Pi(\Ip)$
is a Poisson point process on $\mathbb{R}^2\times [0, g)$. Its mean measure is given by
\begin{equation}\label{Eq:MMeasure:c}
\mu^*(\mathcal{B}\mid J_P=g) = \lambda\iint\limits_{(t, u)\in\mathcal{B}} \mu(dt) p(t, du)
\end{equation}
where $\mathcal{B}$ is a measurable subset of $\mathbb{R}^2\times [0, g)$ and   $p(T, \cdot)$ is the distribution function of $J_T$ conditioned on node $T$.  
\end{lemma}

\begin{lemma} \label{Lem:Primary:PC} The primary pre-cancellation interference power $\Ip$  has the following cumulative distribution function 
\begin{equation}\label{Eq:Ip:CDF}
\Pr(\Ip \leq g) = \sum_{k=0}^{L}\frac{\left(\nu \lambda g^{-\frac{2}{\alpha}}\right)^k}{\Gamma(k+1)}e^{-\nu\lambda g^{-\frac{2}{\alpha}}}
\end{equation}
and the probability density function
\begin{equation}\label{Eq:PDF:G}
f_{P}(g) = \frac{2 (\nu\lambda)^{L+1}}{\alpha \Gamma(L+1)}  g^{- \frac{2(L+1)}{\alpha} -1}e^{-\nu\lambda g^{-\frac{2}{\alpha}}}
\end{equation}
where $\nu = \frac{\pi \Gamma(N+\frac{2}{\alpha})}{\Gamma(N)}$. 
\end{lemma}

After perfect interference cancellation, the interference power of a remaining interferer $T\in\Phi\backslash\mathcal{T}$ is $ I_T = |\bv_0^\dagger \bh_T|^2$.  Define $\tilde{\bh}_T = \bh_T/\|\bh_T\|$ and $\delta_T = |\bv_0^\dagger \tilde{\bh}_T|^2$. Then $I_T$ and $J_T$ are related by $I_T = J_T\delta_T$.   The distributions of the random variables $\{\delta_T\}$ are specified in the following lemma. 
\begin{lemma} \label{Lem:Delta}
The set of random variables $\{\delta_T\mid T \in\Phi\backslash\mathcal{T}\}$ follow i.i.d. beta$(1, N-1)$ distributions as specified by the following probability density  function 
\begin{equation}
f_{\delta}(x) = (N-1) (1-x)^{N-2}, \quad x\in(0,1). \label{Eq:PDF:Delta}
\end{equation}
Furthermore, $\delta_T$ is independent with $J_T$. 
\end{lemma}
\begin{proof} See Appendix~\ref{App:Delta}. 
\end{proof}
\noindent It follows that the primary interference power $\Ipp$  is given as   $\Ipp = \Ip\ddp$ where $\ddp$ is a beta$(1, N-1)$ random variable. Given that $\ddp$ is random, node $T\sp$ may not contribute the largest interference power despite its dominance over other uncanceled interferers in terms of  pre-cancellation interference power. Next, the total secondary interference power can be written as $I_S = \sum_{T\in\Phi\backslash(\mathcal{T}\cup \{T\sp\})}J_T\delta_T$. This summation is a type of   \emph{shot noise} \cite{Kingman93:PoissonProc} whose probability density function  has no known closed-form expression except for some simple cases \cite{Baccelli:AlohaProtocolMultihopMANET:2006, WeberAndrews:TransCapAdHocNetwkDistSch:2006}. Nevertheless, the  conditional  first and second moments  of $I_S$ can be obtained as shown in the following lemma.
\begin{lemma}\label{Lem:InterfPwr:l} After perfect spatial interference cancellation, the secondary interference power $I\ss$  conditioned on the primary pre-cancellation interference power $\Ip = g$  has the following mean and variance  
    \begin{eqnarray}
\E[I\ss\,|\, \Ip = g] &=&     \frac{2\nu\lambda}{N(\alpha-2)}g^{1-\frac{2}{\alpha}}\label{Eq:ExpInterfPwr}\\
\var(I\ss \,|\, \Ip = g) &=& \frac{2\nu\lambda}{N(N+1)(\alpha-1)}g^{2-\frac{2}{\alpha}}.\label{Eq:IPwrSec:Var}
\end{eqnarray}
\end{lemma}
\begin{proof}
See Appendix~\ref{App:InterfPwr:l}.
\end{proof}

\subsection{Bounds on  Outage Probability}\label{Section:Pout:Bounds:PerfCSI}
From \eqref{Eq:SIR_a}, the outage probability $\Pout$  can be written as 
\begin{equation}
\Pout=\Pr(I\ss + \Ipp >  W\theta^{-1}). \label{Eq:Pout:Exact}
\end{equation}
It follows that $\Pout$ can be lower bounded as 
\begin{equation}
\Pout \geq \Pr(\Ipp > W\theta^{-1})\label{Eq:PoutLB:PerfCSI}
\end{equation}
by considering only the primary interferer, which is  tight if the primary interfererence is dominant. Using $\Ipp = \Ip\delta\sp$ and \eqref{Eq:Pout:Exact}, an upper bound on $\Pout$ is obtained  as 
\begin{eqnarray}
\Pout  &\leq & \Pr(I\ss + \Ip > W\theta^{-1})\label{Eq:Pout:0}\\
&=&  \Pr(\Ip > W\theta^{-1}) + \Pr(I\ss > W\theta^{-1} - \Ip \mid \Ip \leq W\theta^{-1}) 
\Pr(\Ip \leq W\theta^{-1}) \label{Eq:Pout:UB}
\end{eqnarray}
where~\eqref{Eq:Pout:0} holds since $\delta\sp\leq 1$. 
The above upper bound can be further bounded by applying the following Chebyshev's inequality: 
\begin{equation}
\Pr(I_S \geq a\mid \Ip =g) \leq \min\left\{\frac{\var(I_S \,|\, \Ip=g)}{\left\{a -\E\left[I_S \,|\, \Ip=g\right]\right\}^2}, 1\right\}, \quad \forall \ a  > \E\left[I_S \,|\, \Ip=g\right]
\label{Eq:Chebyshev}.
\end{equation}
Based on \eqref{Eq:PoutLB:PerfCSI}, \eqref{Eq:Pout:UB} and \eqref{Eq:Chebyshev}, bounds on  the outage probability are derived as shown in the following lemma.
\begin{lemma}\label{Lem:PoutBnds}For perfect spatial interference cancellation, the outage probability satisfies 
$\Pout^{\ell}(\lambda) \leq \Pout
\leq \Pout^u(\lambda)$ where:
\begin{enumerate}
\item The lower bound is 
\begin{equation}\label{Eq:Pout:PerCSI:LB}
\Pout^\ell(\lambda) = 1 -\sum_{k=0}^L  \frac{\l(\nu\lambda\theta^{\frac{2}{\alpha}}\r)^k}{\Gamma(k+1)} \E\l[\l(\frac{W}{\delta\sp}\r)^{-\frac{2k}{\alpha}}e^{-\nu\lambda\theta^{\frac{2}{\alpha}}\l(\frac{W}{\delta\sp}\r)^{-\frac{2}{\alpha}}}\r].
\end{equation}

\item Define the subsets $\mathcal{D}_1$ and $\mathcal{D}_2$ of the product space $\mathds{R}^+\times\mathds{R}^+$ as 
\begin{eqnarray}
\mathcal{D}_1  &=& \left\{ (w, g) \mid  w\theta^{-1} -  \E[I\ss \mid J_P = g]\leq g \leq w\theta^{-1} \right\}\\
\mathcal{D}_2  &=& \left\{ (w, g) \mid  g <  w \theta^{-1} -  \E[I\ss\mid  J_P = g]  \right\}. \label{Eq:D2}
\end{eqnarray}
The upper bound is 
\begin{align}\label{Eq:Pout:PerCSI:UB}
\Pout^u(\lambda) = \Lambda_1 + \Lambda_2 + \Lambda_3
\end{align}
where
\begin{align}
  \label{eq:1}
  \Lambda_1 &= 1 -\sum_{k=0}^L  \frac{\l(\nu\lambda\theta^{\frac{2}{\alpha}}\r)^k}{\Gamma(k+1)} \E\l[W^{-\frac{2k}{\alpha}}e^{-\nu\lambda\theta^{\frac{2}{\alpha}}W^{-\frac{2}{\alpha}}}\r] \\
\Lambda_2 &= \iint\limits_{(w, g)\in\mathcal{D}_1}f_W(w)f\sp(g)dwdg \label{eq:2}\\
\Lambda_3 &= \iint\limits_{(w, g)\in\mathcal{D}_2}\min\left\{\frac{\var(I\ss \,|\, \Ip=g)}{\left\{w\theta^{-1} - g -\E\left[I\ss \,|\, \Ip=g\right]\right\}^2}, 1\right\} f_W(w)f\sp(g)dwdg \label{eq:3}
\end{align}
 with $\E(I_S \,|\, \Ip=g)$ and $\var(I_S \,|\, \Ip=g)$ given in Lemma~\ref{Lem:InterfPwr:l}. 
\end{enumerate}
\end{lemma}
\begin{proof}
See Appendix~\ref{App:PoutBnds}.
\end{proof}

\subsection{Asymptotic Transmission Capacity}\label{Section:AsymTxCap:PerCSI}
Using the upper and lower bounds described in Lemma~\ref{Lem:PoutBnds}, the TC 
scaling is analyzed for a large number of  canceled interferers per node ($L\rightarrow \infty$) or varnishing  outage probability ($\epsilon\rightarrow 0$) as follows.

Increasing the number of antennas at each receiver allows more interferers to be canceled, leading to higher TC. The TC scaling as  $L\rightarrow\infty$ is  given in the following theorem.
\begin{theorem}\label{Theo:TxCap:LarAnt} With perfect CSI and fixed array gain $(N-L)$, if the number of canceled interferers per node $L$ is sufficiently large, the transmission capacity is bounded as 
\begin{equation}\label{Eq:TxCap:LargeAnt} 
\frac{1}{\pi}\l[\frac{\epsilon\l(\alpha-2\r)}{2\theta\E[W^{-1}]}\r]^{\frac{2}{\alpha}} \leq \frac{C(L)}{(1-\epsilon)\log_2(1+\theta)L^{1-\frac{2}{\alpha}}} \leq \frac{2}{\pi}\l[\frac{\E[W]}{\theta(1-\epsilon)}\r]^{\frac{2}{\alpha}}.
\end{equation} 
\end{theorem}
\begin{proof}
See Appendix~\ref{App:TxCap:Markov}.
\end{proof}
The relationship in \eqref{Eq:TxCap:LargeAnt} shows that TC grows  in the order of $L^{1-2/\alpha}$, where the growth is faster for steeper path loss.  Intuitively, if interference decays more quickly with distance, canceling the strongest interferers reduces interference more significantly. 

Small target outage probability results in a network of sparse transmitting nodes (i.e.,  $\lambda\rightarrow 0$). For such a sparse network, the relationship between the outage probability and node density is given in the following lemma.
\begin{lemma}\label{Lem:AsymPout:PerfCSI}
With perfect CSI and $\lambda\rightarrow 0$, the outage probability is bounded as follows. 
\begin{enumerate}
\item If $L+1 \leq \alpha$, for sufficiently small $\lambda$, 
\begin{equation}\label{Eq:PoutScale:a}
\kappa_1\leq \frac{\epsilon}{\lambda^{L+1}}\leq  \kappa_2
\end{equation}
where 
\begin{eqnarray}
\kappa_1 &=& \frac{E\l[\delta\sp^{\fa(L+1)}\r]E\l[W^{-\fa(L+1)}\r]\l(\nu\theta^{\frac{2}{\alpha}}\r)^{L+1}}{\Gamma(L+2)}\label{Eq:Kappa1}\\
\kappa_2 &=& \frac{2^{\frac{2}{\alpha}(L+1)+1}\E\l[W^{-\frac{2}{\alpha}(L+1)}\r](\nu\theta^{\frac{2}{\alpha}})^{L+1}}{\Gamma(L+2)}. \label{Eq:Kappa2}
\end{eqnarray}
\item If $L+1> \alpha$, for sufficiently small $\lambda$,
\begin{equation}\label{Eq:PoutScale:b}
\kappa_1\leq \frac{\epsilon}{\lambda^{L+1}} \quad  \textrm{and}\quad \frac{\epsilon}{\lambda^{\alpha}}\leq  \kappa_3
\end{equation}
where 
\begin{equation}
\kappa_3 = \frac{8\theta^2\nu^\alpha \Gamma(L-\alpha+2)\E[W^{-2}]}{N(N+1)(\alpha-1)\Gamma(L+1)}. \label{Eq:Kappa3}
\end{equation}
\end{enumerate}
\end{lemma}
\begin{proof}
See Appendix~\ref{App:AsymPout:PerfCSI}.
\end{proof}
Note that the ratio $\frac{\kappa_2}{\kappa_1}$ decreases as $L$ becomes smaller. This suggests that the asymptotic bounds are tighter for smaller values of $L$.

Using Lemma~\ref{Lem:AsymPout:PerfCSI} and the TC definition  in \eqref{Eq:TxCap}, 
we have the following TC scaling.
\begin{theorem}\label{Theo:TxCap}
With perfect CSI and small target outage probability $\epsilon \rightarrow 0$, the TC is bounded  as follows. 
\begin{enumerate}
\item If $L+1\leq \alpha $, for sufficiently small $\epsilon$, 
\begin{equation}\label{Eq:CapLaw:a}
\kappa_2^{-\frac{1}{L+1}}\leq  \frac{C(\epsilon)}{\log(1+\theta)\epsilon^{\frac{1}{L+1}}}  \leq \kappa_1^{-\frac{1}{L+1}}
\end{equation}
where $\kappa_1$ and $\kappa_2$ are given in Lemma~\ref{Lem:AsymPout:PerfCSI}.
\item If $L+1>\alpha$, for sufficiently small $\epsilon$, 
\begin{equation}\label{Eq:CapLaw:b}
\kappa_3^{-\frac{1}{\alpha}}\leq \frac{C(\epsilon)}{\log_2(1+\theta)\epsilon^{\frac{1}{\alpha}}}  \quad \textrm{and}\quad  \frac{C(\epsilon)}{\log_2(1+\theta)\epsilon^{\frac{1}{L+1}}}\leq \kappa_1^{-\frac{1}{L+1}}
\end{equation}
where $\kappa_3$ is given in Lemma~\ref{Lem:AsymPout:PerfCSI}.
\end{enumerate}
\end{theorem}

The above theorem shows that as $\epsilon$  decreases, $C(\epsilon)$ follows a power law. For $L +1> \alpha$, only bounds on the exponent are known. 
The derivation of the exact scaling  for $L+1>\alpha$ requires a tighter upper bound on outage probability than that based on Chebyshev's inequality  in \eqref{Eq:Chebyshev}. This may require analyzing the distribution function of the secondary interference power, which, however, has no known closed-form expression for the present case.

For $L +1\leq \alpha$, the exponent of the TC power law is $1/(L+1)$.
This power law indicates that $L$ determines the sensitivity of TC to a change in the outage constraint. To facilitate our discussion, rewrite the scaling in Theorem~\ref{Theo:TxCap} as $C(\epsilon) \cong \alpha \epsilon^{\frac{1}{L+1}}$ where ``$\cong$" represents asymptotic equivalence for $\epsilon \rightarrow 0$. The  sensitivity of TC towards changes of the outage constraint decreases inversely with the number of canceled interferers. Reducing the outage probability by two orders of magnitude decreases  the  TC by $10$, $3.2$, and $1.8$-fold in the case of $1$, $3$ and $7$ canceled interferers per node, respectively. 
Last, from simulation results in Section~\ref{Section:Simualtion}, the TC scaling in Theorem~\ref{Theo:TxCap} is observed to also hold for outage probabilities of practical interest ($\epsilon \leq 0.1$). 

{
According to Theorem~\ref{Theo:TxCap}, the decay rate  of the TC with varnishing outage probability can be  slowed down by employing more antennas for interference cancellation  at the cost of increasing CSI estimation overhead,  which can be considered as overhead for local coordination among nearby nodes. However,  to guarantee nonzero network capacity for zero outage probability,  
perhaps the better choice is to rely on centralized scheduling as in \cite{GuptaKumar:CapWlssNetwk:2000}. Nevertheless, such scheduling  requires global coordination and potentially incurs much higher overhead than combining the random access protocol and  spatial interference cancellation. Thus the current setup   balances   network performance  and overhead. }

\section{ Transmission Capacity with Imperfect CSI}\label{Section:Outage:ImperCSI}
{

This  section addresses the effect of imperfect CSI. First, consider the scenario where the network remains unchanged except that the CSI is imperfect and
the users have to relax their quality-of-service (QoS) requirements,  namely to tolerate higher outage probability represented by $\tPout$ and to lower the date rate from $\log_2(1+\theta)$ to $\log_2(1+\tilde{\theta})$, where  $\tilde{\theta}$ denote the corresponding  SIR threshold.  Thus $\tPout = \Pr(\widetilde{\SIR} \leq \tilde{\theta})$ where $\widetilde{\SIR}$ is given in \eqref{Eq:SIR_b}. The corresponding TC is  
\begin{equation}\label{Eq:TXCap:ImpCSI}
\tilde{C} = (1-\tPout) \log_2(1+\tilde{\theta})\lambda. 
\end{equation}
It is interesting to investigate the required training sequence length under a constraint on the QoS degradation. To this end, define 
\begin{align}
\Delta P &= \tPout(\tilde{\theta}) - P(\theta)\\
\Delta B &= \log_2(1+\theta) - \log_2(1+\tilde{\theta})\\
\Delta C &= C - \tilde{C}\label{Eq:CapDiff}.
\end{align}
We consider the following constraints on the QoS degradation: $\Delta P \leq \vartheta_p$ and $\Delta B\leq \vartheta_b$ with $\vartheta_p, \vartheta_b\geq 0$. The training sequence length that satisfies these constraints and  the corresponding TC loss, specified by the ratio $\frac{\Delta C}{C}$,  are shown  in the following theorem. 
\begin{theorem}\label{Theo:PoutBnds:ImpCSI} To satisfy the constraints $\Delta P \leq \vartheta_p$ and $\Delta B\leq \vartheta_b$, it is sufficient to choose the training-sequence length as 
\begin{equation}\label{Eq:TrainLen}
M = \max\left(\l\lceil \frac{\log L - \log\vartheta_p}{\omega (2^{\vartheta_b}-1)}
\r\rceil, L\right)
\end{equation}
with $\omega = \l[\Gamma(L+1)\r]^{-\frac{1}{L}}$. Moreover, the resultant TC loss normalized by $C$ is bounded as
\begin{equation}\label{Eq:CapLoss}
\frac{\Delta C}{C} \leq \frac{\vartheta_p}{1-\epsilon} + \frac{\vartheta_b}{\log_2(1+\theta)}. 
\end{equation}
\end{theorem}
\begin{proof}
See Appendix~\ref{App:PoutBnds:ImpCSI}.
\end{proof}
Note that \eqref{Eq:TrainLen} takes into account that $M\geq L$ (see Section~\ref{Section:IntfAvoid:ImpCSI}) and $M$ is an integer. For stringent constraints $\vartheta_p \rightarrow 0$ and $\vartheta_b\rightarrow 0$, it can be observed from \eqref{Eq:TrainLen} that the training sequence length is approximately proportional to $\log\frac{1}{\vartheta_p}$ and $\frac{1}{\vartheta_b}$. Also, the upper-bound in \eqref{Eq:CapLoss}  suggests that the normalized capacity loss is more sensitive to the variation of  $\vartheta_p$  if  $\epsilon$ is large and $\vartheta_b$ if  $\log_2(1+\theta)$ is small.

Next, for small outage probability, the required training sequence length is derived for achieving the same TC scaling as for perfect CSI given identical QoS requirements. These results are shown in the following theorem. 
\begin{theorem}\label{Theo:TxCap:ImpCSI}
For $\epsilon \rightarrow 0$ and $\Delta B \overset{\epsilon}\rightarrow 0$, $\Delta P\overset{\epsilon}{\rightarrow} 0$,\footnote{The notation $A\overset{\epsilon}{\rightarrow} B$ represents the convergence $A\rightarrow B$ as  $\epsilon \rightarrow 0$.}, the following scaling of the training sequence length
\begin{equation}
\lim_{\epsilon\rightarrow 0 }\frac{M}{\frac{1+\varrho}{\omega}\epsilon^{-\varrho}}=1\label{Eq:M:Scale}
\end{equation}
with an arbitrary $\varrho > 0$ is sufficient for achieving the transmission-capacity scaling for perfect CSI as given in Theorem~\ref{Theo:TxCap}.
\end{theorem}
\begin{proof}
See Appendix~\ref{App:TxCap:ImpCSI}.
\end{proof}
The above theorem shows that to achieve the optimal asymptotic TC, $M$ increases slowly (sub-linearly with an arbitrary positive exponent) with $\frac{1}{\epsilon}$, indicating small CSI estimation overhead. The reason is that CSI inaccuracy rises mainly from weak interferers and thus its effect is moderate.
}

Suppose the CSI estimation  process repeats for every channel coherence time $t_c$ (in symbols).  The overhead of CSI estimation can be regarded as the TC decrease by a factor of $M/t_c$.
Simulation results in Fig.~\ref{Fig:TxCap:CSI}  reveal that the capacity gain from interference cancellation   results mostly  from canceling only a few strongest interferers and thus $L$ can be kept small; furthermore, short training sequences (small $M$) are sufficient for approaching the TC achieved with  perfect CSI. 
Thus, the overhead is insignificant
if mobility is low (large $t_c$).


\section{Simulation and Discussion}\label{Section:Simualtion}
In this section, the bounds on  outage probability and TC are evaluated using Monte Carlo simulation. The procedure for simulating a MANET follows that in \cite{WeberKam:CompComplexMANETs:2006}. The simulated ad hoc network lies on a two-dimensional disk and contains a number of transmitter-receiver pairs, which is a Poisson random variable with the mean equal to $200$. 
The disk area is adjusted according to the node density. The typical receiver is placed at the center of the disk. We set the required SIR as $\theta = 3$ or $4.8$ dB, the link array gain $(N-L) = 2$, and the path-loss exponent as  $\alpha = 4$ unless specified otherwise.

\subsection{Bounds on Outage Probability}\label{Section:Sim:Pout}

\begin{figure}
\centering
\includegraphics[width=10cm]{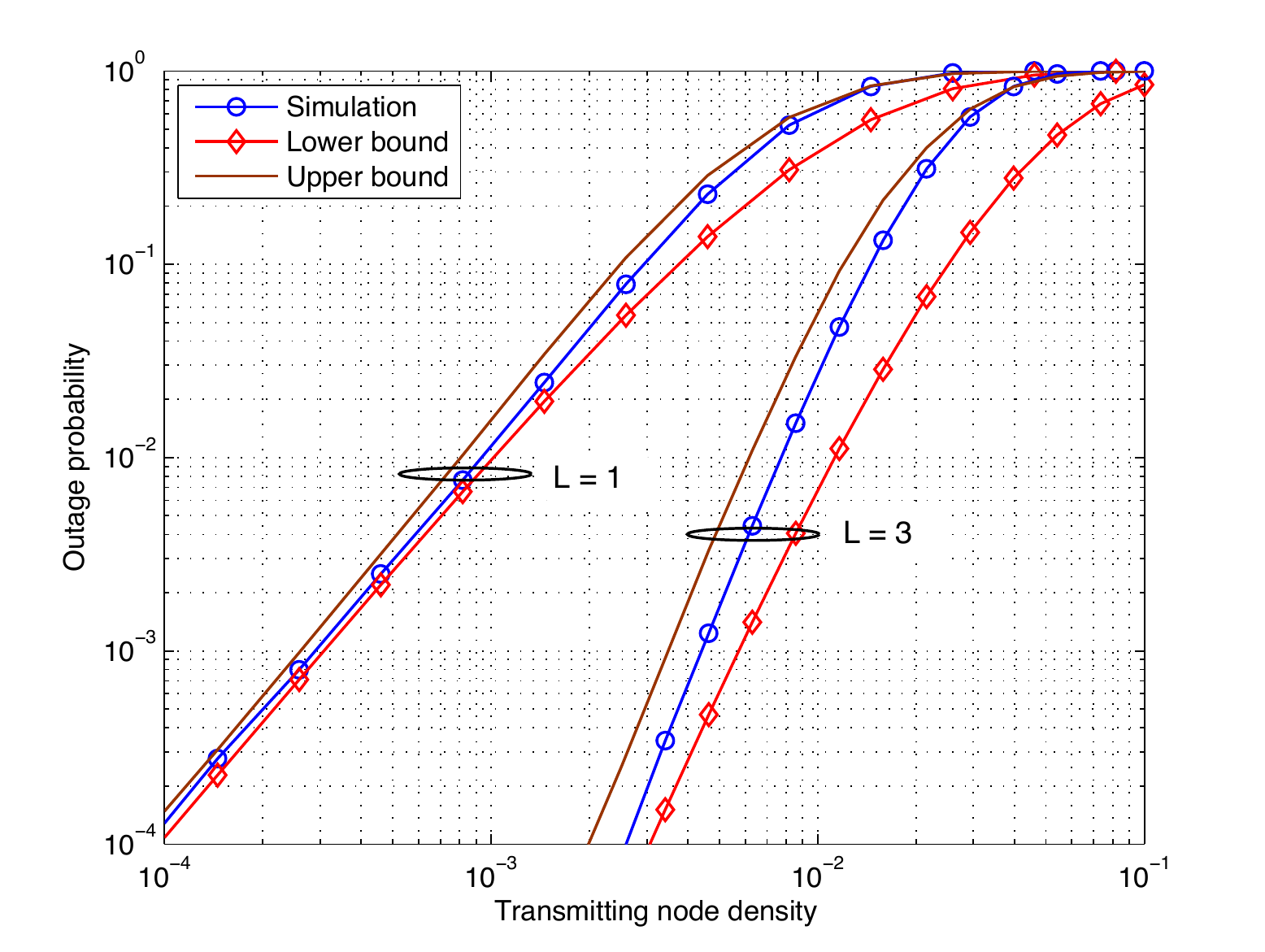}
  \caption{Outage probability for different  transmitting node densities and perfect CSI.  } \label{Fig:Pout:PerfCSI}
\end{figure}

For perfect CSI, the bounds on  outage probability from Lemma~\ref{Lem:PoutBnds} and simulated values are compared in Fig.~\ref{Fig:Pout:PerfCSI}.
It can be observed that
the outage probability is approximately proportional to $\lambda^{L+1}$.
The bounds  for $L=1$ are tighter than those for $L=3$.
Moreover, the bounds on outage probability converge to the exact values as the transmitting node density $\lambda$ decreases. These two observations can be explained by the dominance of the primary interference over the secondary one  as $L$  or $\lambda$ decreases, where the secondary interference causes the looseness of the bounds on outage probability.

{
In Fig.~\ref{Fig:Pout:ImpCSI}, given identical data rates ($\Delta B = 0$),  the outage probability  for imperfect CSI is observed to rapidly converge to the perfect-CSI counterpart as $M$ increases. In particular, for $M=11$, CSI inaccuracy increases outage probability by less than two-fold. 
}

\begin{figure}[t]
\centering
\includegraphics[width=9cm]{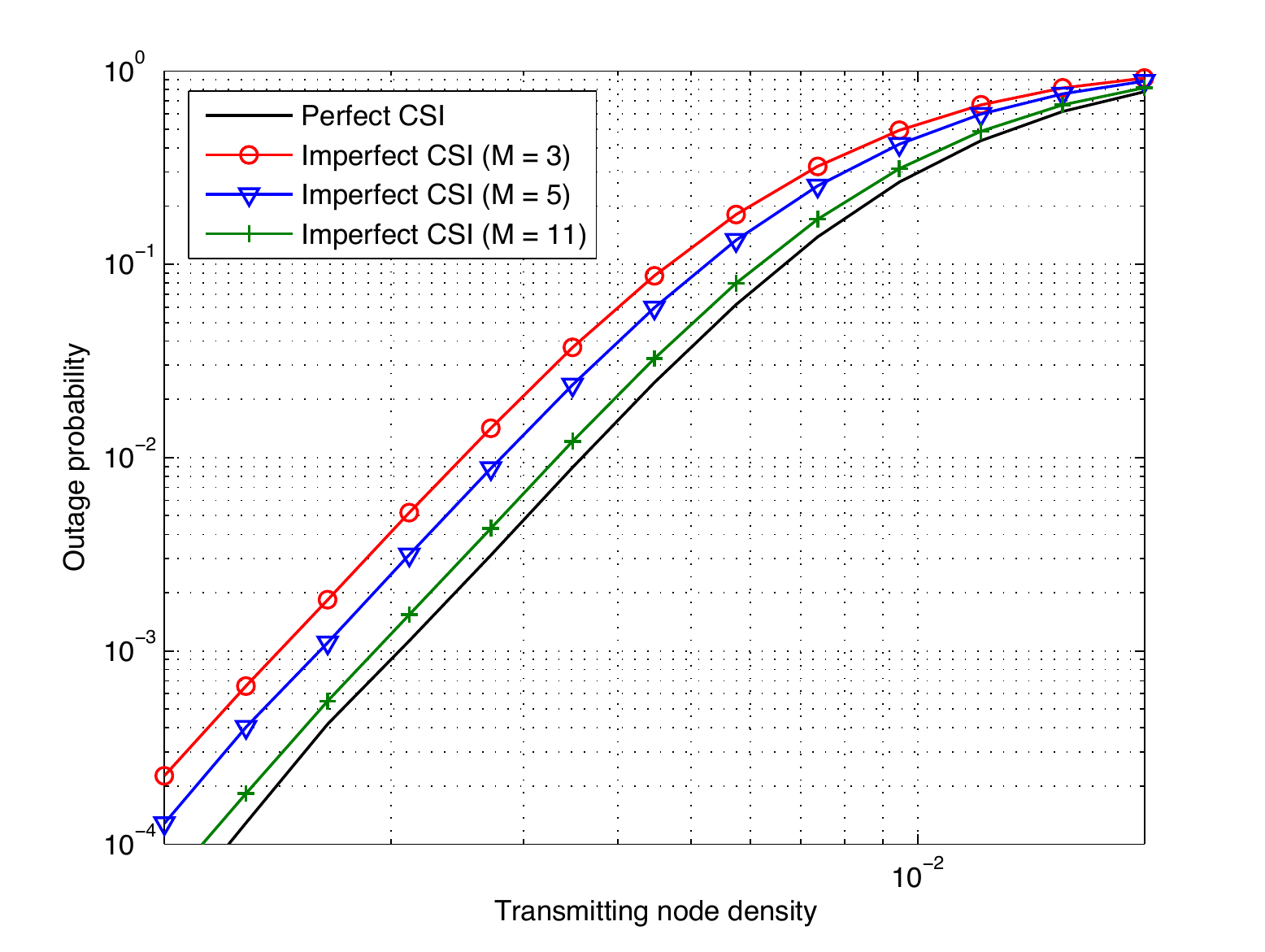}\\
  \caption{Compare outage probability for perfect and imperfect CSI given different  transmitting node densities. The training-sequence length is $M=\{3, 5, 11\}$ and the number of canceled interferers per node is $L=3$. }\label{Fig:Pout:ImpCSI}
\end{figure}

\subsection{Scaling of Transmission Capacity}\label{Section:Sim:TxCapScale}

In Fig.~\ref{Fig:TxCapScale}, asymptotic bounds on TC in Theorem~\ref{Theo:TxCap} are compared with the exact values obtained by simulation for perfect CSI and the range of target outage probability  $\epsilon \in [10^{-5}, 10^{-1}]$. The corresponding curves are identified using the legends ``asymptotic upper bound", ``asymptotic lower bound", and ``simulation". Different combinations of $(L, \alpha)$ are separated according to the cases of $L+1\leq \alpha$ and $L+1> \alpha$, corresponding to  Fig.~\ref{Fig:TxCapScale}(a) and Fig.~\ref{Fig:TxCapScale}(b), respectively. As observed from  Fig.~\ref{Fig:TxCapScale}(a), for $L+1\leq \alpha$,  the asymptotic upper bound on TC is tight even in the non-asymptotic range e.g.,  $\epsilon \in [0.01, 0.1]$. The tightness of this bound is due to the dominance of primary interference when $L$ is small. Moreover, Fig.~\ref{Fig:TxCapScale}(b) shows that for $L+1 >\alpha$ the slopes of the ``simulation" curves converge to those of the corresponding ``asymptotic upper bound" curves as the target outage probability decreases. The above observations suggest that for both $L+1\leq \alpha$ and $L+1> \alpha$, the TC scaling  for small target outage probability follows the power law with the same exponent $\frac{1}{L+1}$.

\begin{figure}[t]
\centering
\hspace{-20pt}\subfigure[$L+1\leq \alpha$]{  \includegraphics[width=9cm]{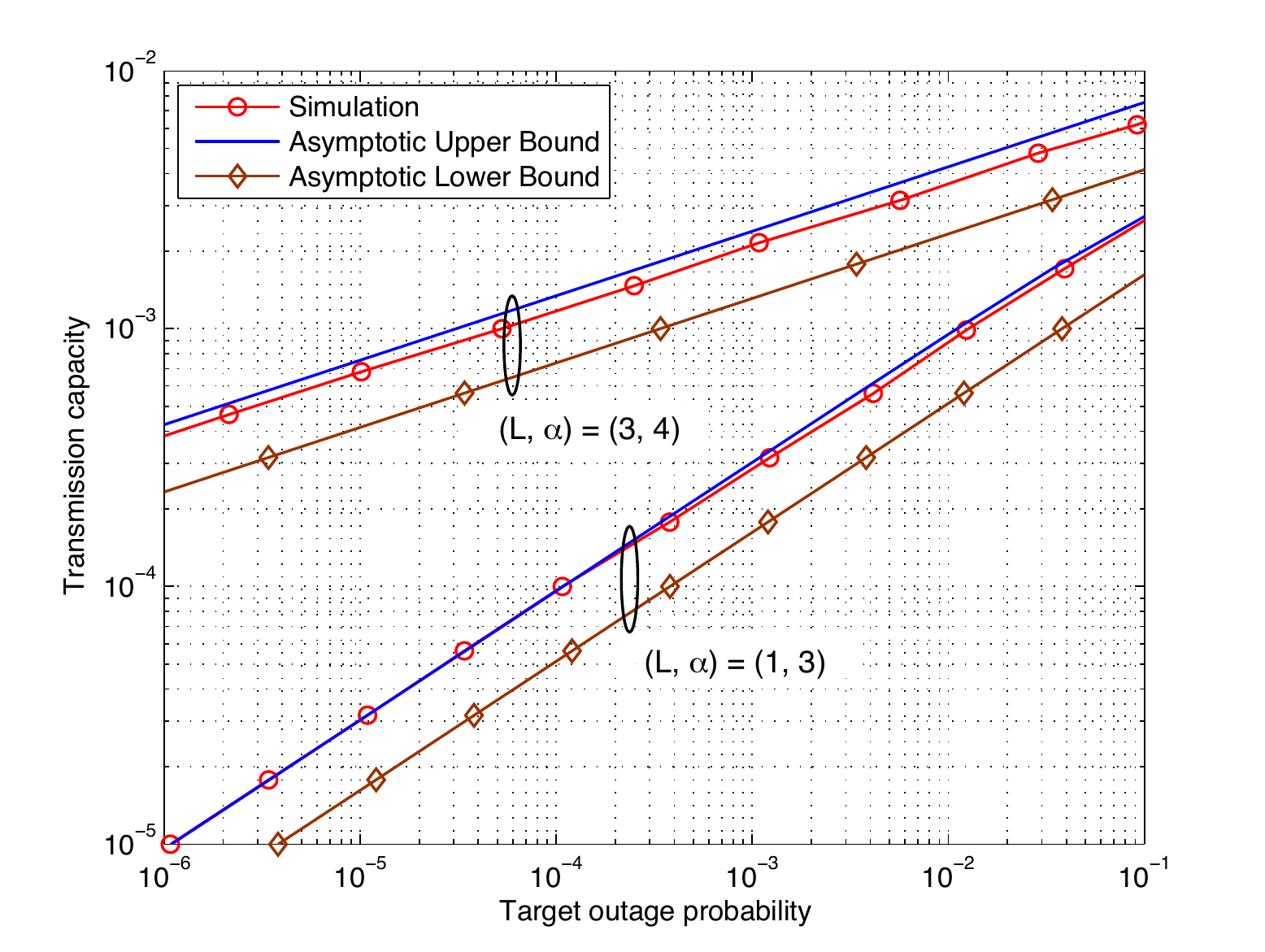}\hspace{-20pt}}
\subfigure[$L+1> \alpha$]{  \includegraphics[width=9cm]{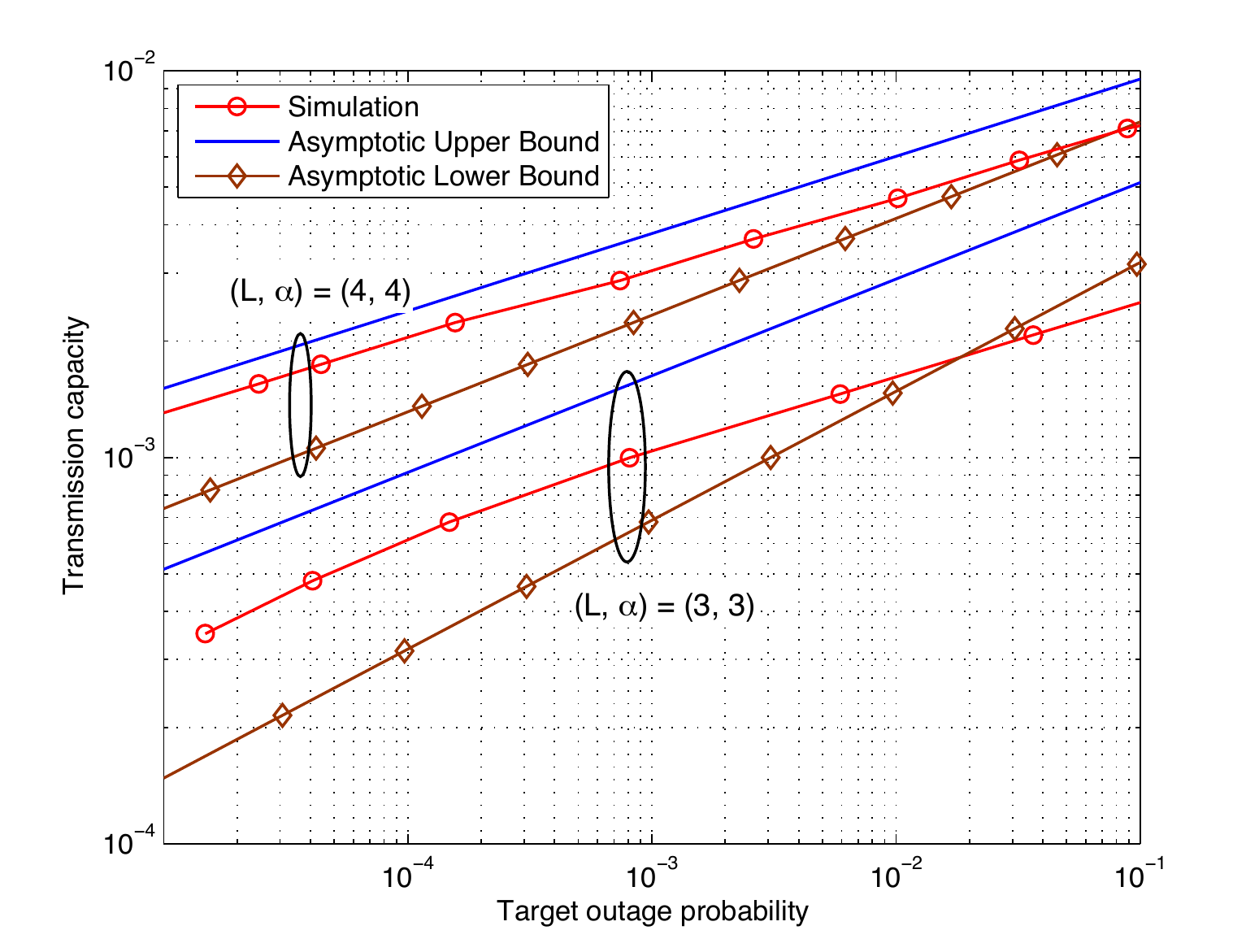}}\hspace{-20pt}\\
  \caption{Comparison between asymptotic bounds on TC and the exact values obtained by simulation for
  perfect CSI and the cases of (a) $L+1 \leq  \alpha$ and (b) $L+1 > \alpha$.}\label{Fig:TxCapScale}
\end{figure}

\subsection{Transmission Capacity vs. Size of Antenna Array}\label{Section:Sim:TxCap}
In Fig.~\ref{Fig:TxCap:Eps}, the transmission capacity is plotted for an increasing number of canceled interferers per node assuming perfect CSI. Furthermore, different target outage probabilities, namely $\epsilon = \{10^{-1}, 10^{-2}, 10^{-3}\}$, are considered.
As observed from Fig.~\ref{Fig:TxCap:Eps}, the cancellation of a few interferers by each node leads to a TC gain of an order of magnitude or more with respect to the case of no cancellation. For example, for $\epsilon = 10^{-2}$, canceling two interferers  per node provides a $25$-time TC gain. The cancellation of more interferers has a diminishing  effect on the network capacity since it becomes limited by secondary  interference. It is also observed that the outage constraint affects TC more significantly for smaller $L$. 

\begin{figure}
\centering
\includegraphics[width=10cm]{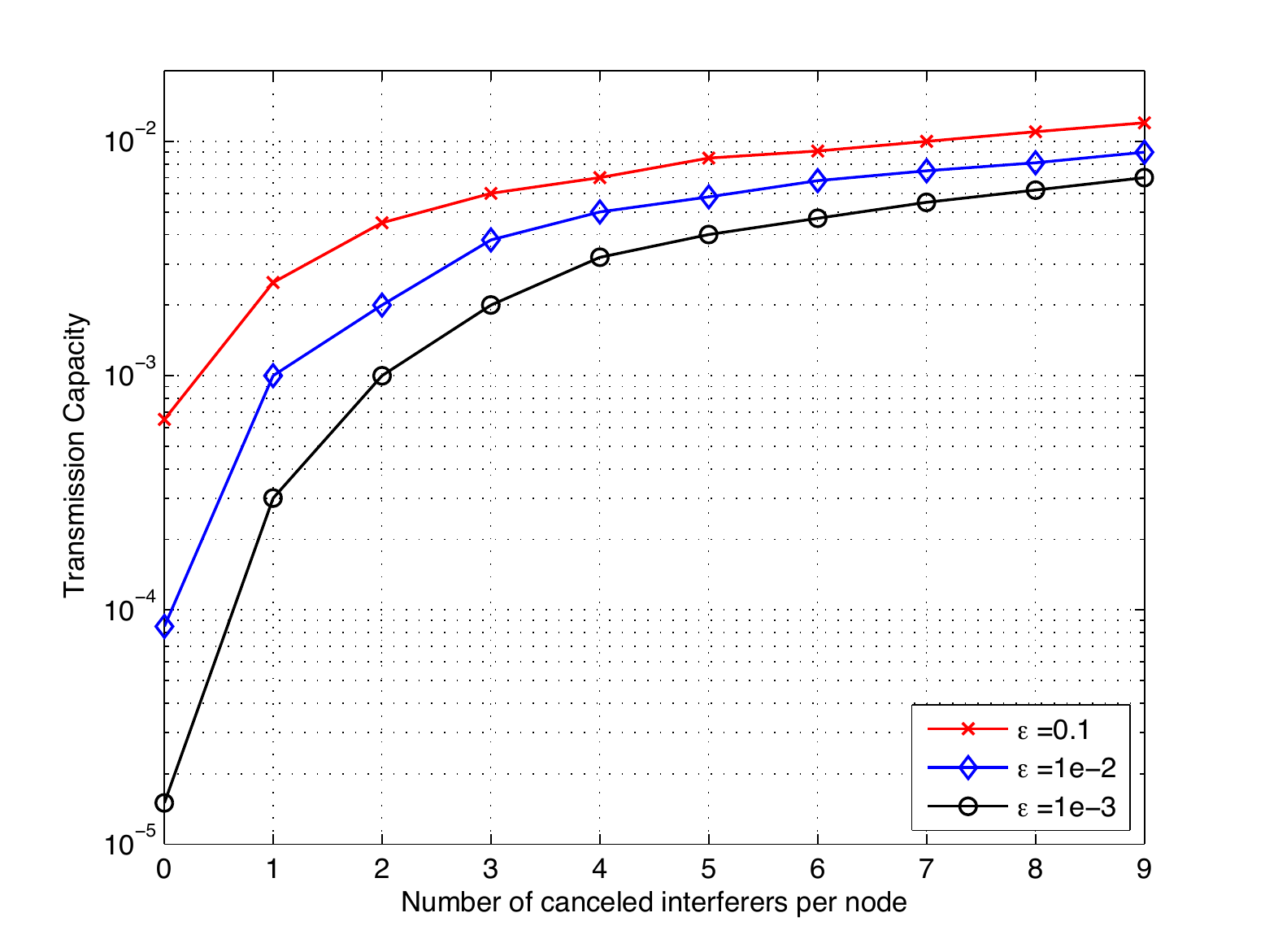}
  \caption{Transmission capacity by simulation for different  node densities and perfect CSI. The target outage probability is $\epsilon = \{10^{-1},10^{-2}, 10^{-3}\}$. }\label{Fig:TxCap:Eps}
\end{figure}

The effect of imperfect CSI on TC is shown in Fig.~\ref{Fig:TxCap:CSI}, where TC is plotted for  increasing  $L$.  The TC loss  due to CSI estimation errors is observed to reduce as $M$ increases. Such a loss is relatively small even for a moderate value of $M$. For instance, the TC reduction  is  $25$\% for $M=11$ and $L=7$. Next,  even for
small $M$ (i.e.,  $M= 3$), a TC gain of more than an order of magnitude can be achieved by  interference cancellation. This confirms  the practicality  of  interference cancellation. 
 
\begin{figure}
\centering
\includegraphics[width=10cm]{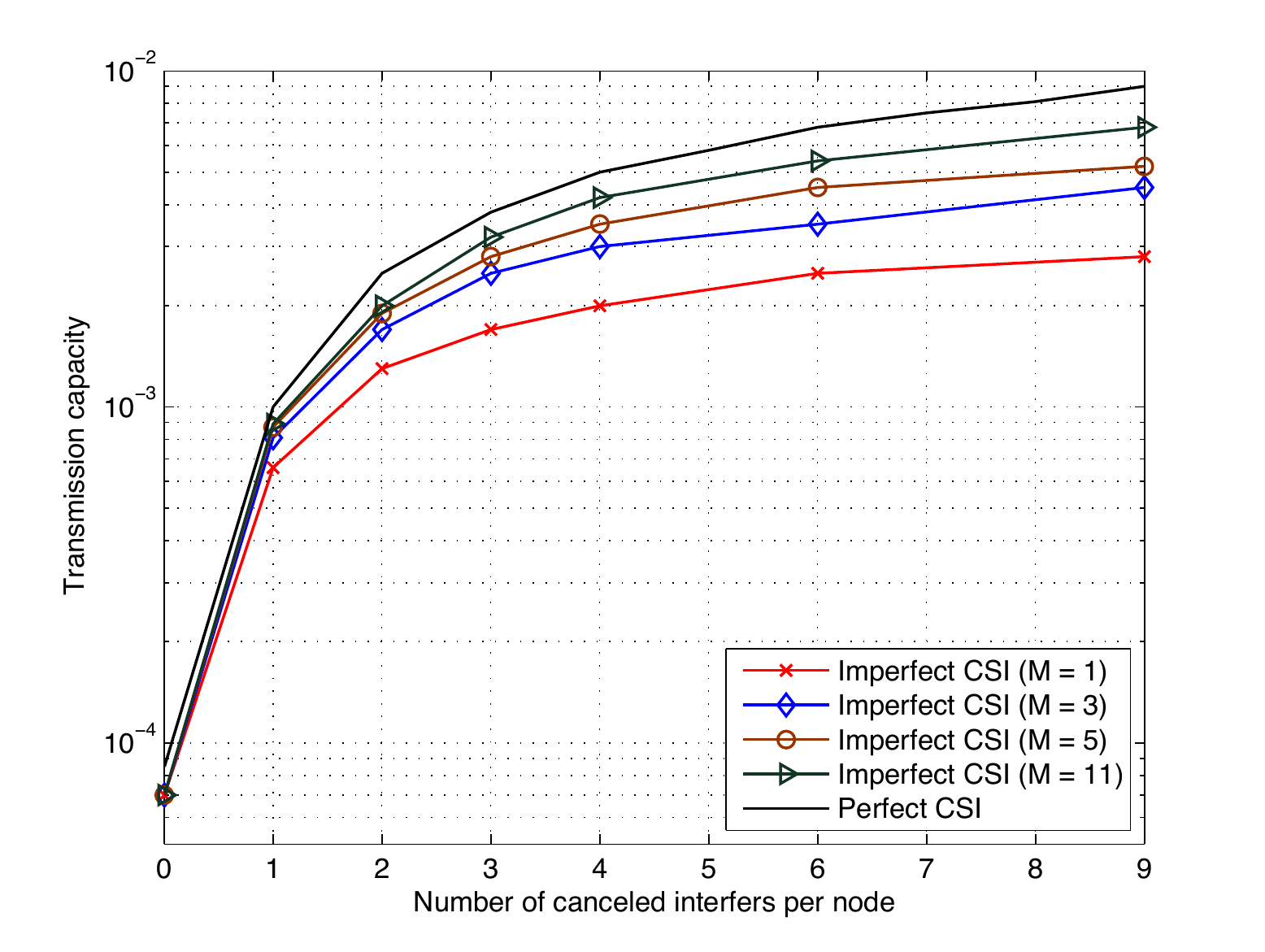}\\
  \caption{Transmission capacity by simulation  for different  node densities and imperfect CSI.  The target outage probability  is $\epsilon = 10^{-2}$. }\label{Fig:TxCap:CSI}
\end{figure}

\section*{Acknowledgement}
The authors thank Nihar Jindal, Rahul Vaze,   and Gustavo de Veciana for helpful discussions. In particular, Dr. Jindal suggested the derivation of Theorem~\ref{Theo:TxCap:LarAnt}.

\appendix

\subsection{Proof for Lemma~\ref{Lem:Second:Dist}}\label{App:Second:Dist}
Consider two disjoint  measurable subsets $\mathcal{B}$ and $\mathcal{C}$ of $\mathbb{R}^2\times [0, g)$. Let $\Xi$ be the counting function such that $\Xi(\mathcal{B})$ gives the number of elements in $\mathcal{B}$.  We first show that $\Xi(\Pi(g)\cap \mathcal{B})$ conditioned on $J_P =g$ is a Poisson random variable as follows. Define the sets $\mathcal{G}_\tau := \mathbb{R}^2\times (g-\tau, \infty)$ and $\mathcal{D}_\tau = \mathbb{R}^2\times [g-\tau, g+\tau]$ where $\tau > 0$.   Then 
\begin{equation}\label{Eq:ProbCount}
\Pr(\Xi(\Pi(g)\cap \mathcal{B}) = n\mid J_P=g) =  \lim_{\tau\rightarrow 0}\Pr(\Xi(\Phi\cap \mathcal{B}) = n \mid \Xi(\Phi\cap \mathcal{D}_\tau) = 1,   \Xi(\Phi\cap \mathcal{G}_\tau) = L). 
\end{equation}
By letting $\tau\rightarrow 0$, the probability measures of $\mathcal{B}\cap \mathcal{D}_\tau$
and $\mathcal{B}\cap \mathcal{G}_\tau$ can be made arbitrarily small.
Since $\Phi$ is Poisson distributed, $\Xi(\Phi\cap \mathcal{B})$ becomes independent with  $\Xi(\Phi\cap \mathcal{D}_\tau)$ and $\Xi(\Phi\cap \mathcal{G}_\tau)$ in the limit of $\tau\to0$. It follows that \eqref{Eq:ProbCount} can be rewritten as 
\begin{equation}\label{Eq:ProbCount:a}
\Pr(\Xi(\Pi(g)\cap \mathcal{B}) = n\mid J_P=g) =  \Pr(\Xi(\Phi\cap \mathcal{B}) = n). 
\end{equation}
In other words, $\Xi(\Pi(g)\cap \mathcal{B})$ conditioned on $J_P=g$ follows the  Poisson distribution with the same parameter as  $\Xi(\Phi\cap \mathcal{B})$. 

Next, we prove the independence between $\Xi(\Pi(g)\cap \mathcal{B})$ and $\Xi(\Pi(g)\cap \mathcal{C})$ given $J_P = g$. Their joint distribution function is written as 
\begin{equation}\label{Eq:ProbCount:b}
  \begin{aligned}
\Pr(\Xi(\Pi(g)\cap \mathcal{B})& = m,   \Xi(\Pi(g)\cap \mathcal{C})=n \mid J_P=g) \\
&= \lim_{\tau\rightarrow 0}\Pr(\Xi(\Phi\cap \mathcal{B}) = m, \Xi(\Phi\cap \mathcal{C})=n \mid \Xi(\Phi\cap \mathcal{D}_\tau) = 1, \Xi(\Phi\cap \mathcal{G}_\tau) = L).     
 \end{aligned}
\end{equation}
Following a similar argument as for getting  \eqref{Eq:ProbCount:a}, we can obtain from \eqref{Eq:ProbCount:b} that 
\begin{eqnarray}
\Pr(\Xi(\Pi(g)\cap \mathcal{B}) = m, \Xi(\Pi(g)\cap \mathcal{C})=n\mid J_P=g) &=&  \Pr(\Xi(\Phi\cap \mathcal{B}) = m, \Xi(\Phi\cap \mathcal{C})=n)\nn\\
 &=&  \Pr(\Xi(\Phi\cap \mathcal{B}) = m)\Pr(\Xi(\Phi\cap \mathcal{C})=n)\label{Eq:ProbCount:c}
\end{eqnarray}
where the last equality is due to a property of the Poisson process $\Phi$ given that $\mathcal{B} \cap \mathcal{C} = \emptyset$. The independence between 
$\Xi(\Pi(g)\cap \mathcal{B})$ and $ \Xi(\Pi(g)\cap \mathcal{C})$ conditioned on $J_P = g$ follows from 
\eqref{Eq:ProbCount:a} and \eqref{Eq:ProbCount:c}. 

Combining above results proves that given $J_P = g$,  $\Pi(g)$ is Poisson distributed on the space $\mathbb{R}^2\times [0, g)$; furthermore, $\Pi(g)$ has the mean measure as given in \eqref{Eq:MMeasure}, which leads to \eqref{Eq:MMeasure:c}. This completes the proof.

\subsection{Proof for Lemma~\ref{Lem:Primary:PC}}\label{App:Primary:PC}

The distribution  of $J_T$ can be analyzed  by writing  $J_T = r_T^{-\alpha}\rho_T$ where $\rho_T = \|\bG_T\bff_T\|^2$. Since $\bG_T$ is an i.i.d. $\mathcal{CN}(0, 1)$ matrix and $\bff_T$ fixed, $\bG_T\bff_T$ is an i.i.d. $\mathcal{CN}(0, 1)$ vector. Hence $\rho_T$ has  the chi-square distribution with $N$ complex degrees of freedom. It follows that 
\begin{equation}\label{Eq:CondProb:I}
\Pr(J_T> g \mid r_T = r) = \int_{r^\alpha g}^\infty  \frac{u^{N-1}}{\Gamma(N)}e^{-u}du. 
\end{equation}
Substituting \eqref{Eq:CondProb:I} into \eqref{Eq:MMeasure:a} gives 
\begin{eqnarray}
\mu^*(\mathcal{G})
&=& 2\pi \lambda \int_{0}^{\infty}\int_{r^\alpha g}^\infty  \frac{ru^{N-1}}{\Gamma(N)}e^{-u}du dr\nn\\
&=& 2\pi \lambda \int_{0}^{\infty}\int_0^{\l(\frac{u}{g}\r)^{\frac{1}{\alpha}}}  \frac{ru^{N-1}}{\Gamma(N)}e^{-u} dr du\nn\\
&=& \frac{\pi\lambda g^{-\frac{2}{\alpha}}}{\Gamma(N)} \int_0^\infty u^{N+\frac{2}{\alpha} - 1}e^{-u}du \nn\\
&=& \nu\lambda g^{-\frac{2}{\alpha}}  \label{Eq:MMeasure:b}
\end{eqnarray}
where $\nu$ is defined in the lemma statement. 
Then the distribution function of $\Ip$ can be written as 
\begin{equation}\label{Eq:Ip:CDF:a}
\Pr(\Ip\leq g) = \Pr(\Xi(\Psi\cap\mathcal{G}) \leq L). 
\end{equation}
Since $\Xi(\Psi\cap\mathcal{G})$ is a Poisson random variable with the mean $\mu^*(\mathcal{G})$ given in \eqref{Eq:MMeasure:b}, the desired cumulative distribution function  in \eqref{Eq:Ip:CDF} follows from \eqref{Eq:Ip:CDF:a}. 
Differentiating this function gives the probability density function of $\Ip$ as 
\begin{eqnarray}
f_{P}(g) &=& \frac{2}{\alpha}\sum_{k=0}^{L}\frac{(\nu\lambda)^{k+1}}{k!}g^{-\frac{2(k+1)}{\alpha}-1}e^{-\nu\lambda g^{-\frac{2}{\alpha}}}-\frac{2}{\alpha}\sum_{k=1}^{L}\frac{(\nu\lambda)^{k}}{(k-1)!}g^{-\frac{2k}{\alpha} -1}e^{-\nu\lambda g^{-\frac{2}{\alpha}}}\nn\\
&=& \frac{2}{\alpha} e^{-\nu\lambda g^{-\frac{2}{\alpha}}}\left\{\sum_{k=0}^{L}\frac{(\nu\lambda)^{k+1}}{k!}g^{-\frac{2(k+1)}{\alpha}-1} -\frac{2}{\alpha}\sum_{k=0}^{L-1}\frac{(\nu\lambda)^{k+1}}{k!}g^{-\frac{2(k+1)}{\alpha} -1}\right\}.\nn
\end{eqnarray}
The desired result in \eqref{Eq:PDF:G} follows from the last equation.

\subsection{Proof for Lemma~\ref{Lem:Delta}}\label{App:Delta}
Consider an arbitrary interferer $T\in\Phi$ before interference cancellation. We can write the effective channel vector as $\bh_T = J_T\tilde{\bh}_T$. The isotropicity of  $\bh_T$ has two consequences: $J_T$ and $\tilde{\bh}_T$ are independent and $\tilde{\bh}_T$ is also isotropic. Recall that $J_T$ is the criterion for selecting interferers to cancel. Thus,  the independence between $J_T$ and $\tilde{\bh}_T$ implies that the isotropicity of $\bh_T$ is unaffected by interference cancellation if node $T$ is uncanceled. 

Next, it can be observed from \eqref{Eq:MRC} that $\bv_0$ is a linear function of the vectors $\bh_0$ and $\{\bh_T\mid T\in\mathcal{T}\}$, which  are i.i.d. and isotropic. As a result, $\bv_0$ is also isotropic as well as independent with other normalized channel vectors $\{\tilde{\bh}_T\mid T\in\Phi\backslash\mathcal{T}\}$. Hence for an uncanceled interferer $T\in\Phi\backslash\mathcal{T}$, $\delta_T = |\bv_0^\dagger\tilde{\bh}_T|^2$ represents the product of two independent isotropic unit-norm random vectors   $\bv_0$ and $\tilde{\bh}_T$, which is shown in  \cite{YeungLove:RandomVQBeamf:05} to have  the beta$(1, N-1)$ distribution  in \eqref{Eq:PDF:Delta}. Moreover, the  independence between $\delta_T$ and $\delta_{T'}$ for $T\neq T'$ follows from the independence between $\tilde{\bh}_T$ and $\tilde{\bh}_{T'}$. This proves the first claim in the lemma statement. 

Last, given $T\in\Phi\backslash\mathcal{T}$, since both $\bv_0$ and $\tilde{\bh}_T$ are independent with $J_T$ as mentioned above, the independence of $\delta_T$ with $J_T$ is immediate. This completes the proof. 

\subsection{Proof of Lemma~\ref{Lem:InterfPwr:l}}\label{App:InterfPwr:l}
Define the sum pre-cancellation secondary interference power as $J\ss := \sum_{T\in\Phi\backslash(\mathcal{T}\cup\{T\sp\})}J_T$. 
Using Lemma~\ref{Lem:Second:Dist}, the application of  Campbell's Theorem gives that  \cite{StoyanBook:StochasticGeometry:95}
\begin{eqnarray}
\E[J\ss\,|\, \Ip=g] &=& \lambda\int_{\mathds{R}^2} \int_0^g u     p(|x|, du ) du dx\nn\\
&=& 2\pi \lambda\int_0^\infty \int_0^g ru     p(r, du ) du dr\nn\\
&=& 2\pi \lambda\int_0^\infty \int_0^{r^\alpha g} r^{1-\alpha} \frac{\rho^N}{\Gamma(N)} e^{-\rho}   d\rho dr\nn\\
&=&2\pi\lambda\int_0^\infty \int_{\left(\frac{\rho}{g}\right)^{\frac{1}{\alpha}}}^\infty r^{1-\alpha}\frac{\rho^N}{\Gamma(N)} e^{-\rho} dr d\rho\nn\\
&=&\frac{2\pi\lambda}{\alpha-2}g^{1-\frac{2}{\alpha}}\int_0^\infty  \frac{\rho^{N+\frac{2}{\alpha}-1}}{\Gamma(N)} e^{-\rho} d\rho\nn\\
&=&\frac{2\nu\lambda}{\alpha-2}g^{1-\frac{2}{\alpha}}. \label{Eq:Is:Mean}
\end{eqnarray}
Next, the expectation of $I\ss$ conditioned on $\Ip = g$ is given as 
\begin{eqnarray}
\E[I\ss\,|\, \Ip=g] &=&  \E\l[\sum_{T\in\Phi\backslash(\mathcal{T}\cup\{T\sp\})}J_T\delta_T\mid \Ip = g\r]\nn\\
&=& \E\l[\sum_{T\in\Phi\backslash(\mathcal{T}\cup\{T\sp\})}J_T\E(\delta_T)\mid \Ip = g\r]\label{eq:58}\\
&=& \frac{1}{N} \E[J\ss\,|\, \Ip=g] \label{Eq:PCI:Mean}
\end{eqnarray}
where~\eqref{eq:58} holds since $\delta_T$ is independent with $\Ip$ according to Lemma~\ref{Lem:Delta},  and~\eqref{Eq:PCI:Mean} uses $\E[\delta_T] = \frac{1}{N}$ derived using the distribution function in \eqref{Eq:PDF:Delta}. Substituting \eqref{Eq:Is:Mean} into \eqref{Eq:PCI:Mean} gives the desired result in \eqref{Eq:ExpInterfPwr}. 

Like \eqref{Eq:Is:Mean}, the conditional variance of $J\ss$ is obtained by applying Campbell's Theorem as follows 
\begin{eqnarray}
\var(J\ss \,|\, \Ip=g) &=& 2\pi\lambda\int_0^\infty\int_{\left(\frac{\rho}{g}\right)^{\frac{1}{\alpha}}}^\infty  r\left(r^{-\alpha}\rho\right)^2  \frac{\rho^{N-1}}{\Gamma(N)}e^{-\rho} d\rho dr\nn\\
&=& \frac{\pi\lambda}{\alpha-1}g^{2-\frac{2}{\alpha}}\int_0^\infty \rho^{N+\frac{2}{\alpha}-1}e^{-\rho} d \rho\nn\\
&=& \frac{\nu \lambda}{\alpha-1}g^{2-\frac{2}{\alpha}}. \label{Eq:Var:Js}
\end{eqnarray}
Since $\delta_T$ is independent with $\Ip$, the conditional variance of $I\ss$ is obtained by applying modified Campbell's Theorem  \cite[p77]{Kingman93:PoissonProc}
\begin{equation}
\var(I\ss \,|\, \Ip=g) = \E[\delta_T^2]\var(J\ss \,|\, \Ip=g).\label{Eq:PCI:Var}
\end{equation}
The second moment  $\E[\delta_T^2]$ can be obtained using the distribution function in \eqref{Eq:PDF:Delta} as 
\begin{eqnarray}
\E[\delta_T^2] &=& (N-1)B(3, N-1)\nn\\
&=& \frac{2}{N(N+1)}\label{Eq:Delta:M2}
\end{eqnarray}
where $B$ denotes the beta function, and \eqref{Eq:Delta:M2} applies the formula $B(x, y) = \frac{\Gamma(x)\Gamma(y)}{\Gamma(x+y)}$ from \cite[8.384]{GradRyzhik:Integral:2007}. Combining \eqref{Eq:Var:Js}, \eqref{Eq:PCI:Var} and \eqref{Eq:Delta:M2} gives the desired result in \eqref{Eq:IPwrSec:Var}. This completes the proof. 

\subsection{Proof for Lemma~\ref{Lem:PoutBnds}}\label{App:PoutBnds}
We can rewrite \eqref{Eq:PoutLB:PerfCSI} as 
\begin{equation}
\Pout \geq 1-\Pr\l(\Ip \leq  \frac{W\theta^{-1}}{\delta\sp}\r).\nn
\end{equation}
Substituting \eqref{Eq:Ip:CDF} into the above equation gives the lower bound on $\Pout$ as shown in 
\eqref{Eq:Pout:PerCSI:LB}. Similarly, we can obtain the first term of the $\Pout$ upper bound in \eqref{Eq:Pout:UB} as 
\begin{equation}\label{Eq:Ip:CCDF}
\Pr(\Ip > W\theta^{-1}) = \Lambda_1 
\end{equation}
where $\Lambda_1$ is defined  \eqref{eq:1}. Given  $\mathcal{D}_1\cup\mathcal{D}_2 = \{(w, g)\mid g < w\theta^{-1}\}$, the second  term of the $\Pout$ upper bound in \eqref{Eq:Pout:UB} can be expanded and then upper bounded as 
\begin{align}
\Pr(I\ss > W\theta^{-1} - &\Ip \mid \Ip \leq W\theta^{-1}) 
\Pr(\Ip \leq W\theta^{-1})\nn\\
\leq& \Pr((W, J_P)\in\mathcal{D}_1) + 
\Pr(I\ss > W\theta^{-1} - \Ip \mid (W, J_P)\in\mathcal{D}_2) \Pr((W, J_P)\in\mathcal{D}_2)\nn\\
\leq&\Lambda_2 + \Lambda_3  \label{Eq:Lambda:23}
\end{align}
where $\Lambda_2$ and $\Lambda_3$ are defined in the lemma statement, and \eqref{Eq:Lambda:23} applies Chebyshev's inequality in \eqref{Eq:Chebyshev}. Combining \eqref{Eq:Pout:UB}, \eqref{Eq:Ip:CCDF} and \eqref{Eq:Lambda:23} gives the desired $\Pout$ upper bound in \eqref{Eq:Pout:PerCSI:UB}. This completes the proof. 

\subsection{Proof of Theorem~\ref{Theo:TxCap:LarAnt}}\label{App:TxCap:Markov}
Given the distribution of $\Ip$ in \eqref{Eq:PDF:G}, 
and applying Campbell's Theorem, we obtain that 
\begin{eqnarray}
\E[\Ipp] &=& \E[\Ip]\E[\delta_P]\nn\\
&=& \frac{2 (\nu \lambda)^{L+1}}{N \alpha\Gamma(L+1)} \int_0^\infty g^{-\frac{2(L+1)}{\alpha}} e^{-\nu\lambda g^{-\frac{2}{\alpha}}} dg\nn\\
&=& \frac{\Gamma\(L+1-\frac{\alpha}{2}\)(\nu\lambda)^{\frac{\alpha}{2}}}{N\Gamma(L+1)}. \label{Eq:Mean:G}
\end{eqnarray}
Moreover, from \eqref{Eq:PDF:G} and \eqref{Eq:ExpInterfPwr}, 
\begin{eqnarray}
\E[I_S] &=& \E[\E[I_S\,|\, \Ip]] \nn\\
&=& \frac{4(\nu\lambda)^{L+2}}{N\alpha(\alpha-2)\Gamma(L+1)} \int_0^\infty g^{-\frac{2(L+2)}{\alpha} }e^{-\nu\lambda g^{-\frac{2}{\alpha}}} dg \label{Eq:Mean:I:a}\\
&=& \frac{2 (\nu\lambda)^{\frac{\alpha}{2}}\Gamma\(L-\frac{\alpha}{2}+2\)}{N(\alpha-2)\Gamma(L+1)}\label{Eq:Mean:I}
\end{eqnarray}
where  \eqref{Eq:Mean:I:a} uses \eqref{Eq:IPwrSec:Var}. 
 Using \eqref{Eq:Mean:G} and \eqref{Eq:Mean:I}, the total interference power $I_\Sigma$ has the following expectation: 
\begin{equation}
\E[I_\Sigma]
= \frac{2(\nu\lambda)^{\frac{\alpha}{2}}}{\alpha-2}\times \frac{L\Gamma\(L-\frac{\alpha}{2}+1\)}{N\Gamma(L+1)}. \label{Eq:SumInterf}
\end{equation}
Using Markov's inequality
\begin{equation}
\Pout \leq \E\l\{ \frac{\E[I_\Sigma]}{W\theta^{-1}}\r\}. \label{Eq:Pout:UB:a}
\end{equation}
Since $\Pout = \epsilon$, it follows from \eqref{Eq:SumInterf} and \eqref{Eq:Pout:UB:a} that 
\begin{equation}
\lambda \geq \frac{1}{\nu}\l\{\frac{N\epsilon\l(\alpha-2\r)}{2L\theta\E[W^{-1}]}\times\frac{\Gamma(L+1)}{\Gamma\l(L-\frac{\alpha}{2}+1\r)}\r\}^{\frac{2}{\alpha}}. \label{Eq:Density:LB}
\end{equation}

Next, a lower bound on $\lambda$ can be derived using the method in \cite{Jindal:RethinkMIMONetwork:LinearThroughput:2008} where the success probability $(1-\Pout)$ is upper bounded using Markov's inequality. To this end, let $\mathcal{U}$ denote the set of $(L+1)$ strongest uncanceled interferers in terms of pre-cancellation interference power. Thus, the weakest interferer in $\mathcal{U}$ corresponds to $\acute{J} = \min_{T\in \mathcal{U}} J_T$. Using above definitions, 
\begin{eqnarray}
1-\Pout &\leq& \Pr\l(\frac{W}{\sum_{T\in\mathcal{U}}J_T\delta_T}\geq \theta\r)\nn\\
&\leq& \Pr\l(\frac{W}{\acute{J}_T\sum_{T\in\mathcal{U}}\delta_T}\geq \theta\r)\nn\\
&\leq& \theta^{-1}\E[W]\E[\acute{J}^{-1}]\E\l[\frac{1}{\sum_{T\in\mathcal{U}}\delta_T}\r]\label{Eq:Success:UB}
\end{eqnarray}
where~\eqref{Eq:Success:UB} uses Markov's inequality. 
Note that the probability density function of $\acute{J}$ is given by \eqref{Eq:PDF:G} with $L$ replaced with $2L$. Thus, similar to \eqref{Eq:Mean:G}, it can be obtained that
\begin{equation}
\E[\acute{J}^{-1}] = \frac{\Gamma(2L+1+\frac{\alpha}{2})}{\Gamma(2L+1)(\nu\lambda)^{\frac{\alpha}{2}}}. 
\end{equation}
Substituting the above equation into \eqref{Eq:Success:UB} gives 
\begin{equation}
\lambda \leq \frac{1}{\nu}\l\{\frac{\E[W]}{\theta(1-\epsilon)}\times\frac{\Gamma(2L+1+\frac{\alpha}{2})}{\Gamma(2L+1)}\times \E\l[\frac{1}{\sum_{T\in\mathcal{U}}\delta_T}\r] \r\}^{\frac{2}{\alpha}}. \label{Eq:Density:UB}
\end{equation}

The scaling of $\lambda$ can be derived using \eqref{Eq:Density:LB} and \eqref{Eq:Density:UB} and applying Kershaw's inequality \cite{Kershaw:GautschiInequGammaFun:1983}:
\begin{equation}\label{Eq:Kershaw}
\l(x +\frac{s}{2}\r)^{1-s} < \frac{\Gamma(x+1)}{\Gamma(x+s)} < \l( x - \frac{1}{2} + \sqrt{s+\frac{1}{4}}\r)^{1-s}, \quad x > 0, \ 0 < s < 1.
\end{equation}
Specifically, given  $\Gamma(1+x) = x\Gamma(x)$, we can write 
\begin{equation}
\frac{\Gamma(L+1)}{\Gamma\(L-\frac{\alpha}{2}+1\)}
= \frac{\Gamma\(L - \lceil\frac{\alpha}{2}\rceil  + 2\)}{\Gamma\(L+1 - \lceil\frac{\alpha}{2}\rceil  + \Delta\alpha \)}\prod_{n=0}^{\lceil\frac{\alpha}{2}\rceil - 2}\(L-n\) \label{Eq:Pout:Markov:a}
\end{equation}
where $\Delta\alpha = \lceil\frac{\alpha}{2}\rceil - \frac{\alpha}{2}$ and hence $0\leq \Delta\alpha< 1$. 
Using Kershaw's inequality, for $L > \l\lceil\frac{\alpha}{2}\r\rceil$
\begin{equation}\label{Eq:Gamma:a}
\(L - \l\lceil\frac{\alpha}{2}\r\rceil+1 + \frac{\Delta \alpha}{2} \)^{1-\Delta\alpha}< \frac{\Gamma\(L - \lceil\frac{\alpha}{2}\rceil  + 2\)}{\Gamma\(L+1 - \lceil\frac{\alpha}{2}\rceil  + \Delta\alpha \)} < \(L - \l\lceil\frac{\alpha}{2}\r\rceil +\frac{1}{2}+ \sqrt{\Delta\alpha + \frac{1}{4}}\)^{1-\Delta\alpha}.
\end{equation}
It follows from  \eqref{Eq:Pout:Markov:a} and \eqref{Eq:Gamma:a} that 
\begin{equation}
\lim_{L\rightarrow\infty}
\frac{\Gamma(L+1)}{L^{\frac{\alpha}{2}}\Gamma\(L-\frac{\alpha}{2}+1\)}=1. \label{Eq:GammaRatio:Limit}
\end{equation}
Similarly, we can show that $\nu$ as defined in Lemma~\ref{Lem:Primary:PC} scales as:
\begin{equation}
\lim_{L\rightarrow\infty}
\frac{\nu}{\pi L^{\frac{2}{\alpha}}}=1. \label{Eq:Nu:Limit}
\end{equation}
Combining \eqref{Eq:Density:LB}, \eqref{Eq:GammaRatio:Limit} and \eqref{Eq:Nu:Limit} gives 
\begin{equation}
\liminf_{L\rightarrow\infty} \frac{\lambda}{L^{1-\frac{2}{\alpha}}} \geq \frac{1}{\pi}\l\{\frac{\epsilon\l(\alpha-2\r)}{2\theta\E[W^{-1}]}\r\}^{\frac{2}{\alpha}}. \label{Eq:LambdaLim:LB}
\end{equation}
Again, the application of Kershaw's inequality yields 
\begin{equation}
\lim_{L\rightarrow\infty}\frac{\Gamma(2L+1+\frac{\alpha}{2})}{(2L)^{\frac{\alpha}{2}}\Gamma(2L+1)} = 1. \label{Eq:GammaRatio:Limit:a}
\end{equation}
Moreover, it follows from the strong law of larger numbers that $\lim_{L\rightarrow\infty}\frac{\sum_{T\in\mathcal{U}}\delta_T}{L+1}=\E[\delta_T] = \frac{1}{L}$. Since $(N-L)$ is fixed,  
\begin{equation}
\lim_{L\rightarrow\infty}\E\l[\frac{1}{\sum_{T\in\mathcal{U}}\delta_T}\r] =1. \label{Eq:SumEps:Limit}
\end{equation}
Substituting \eqref{Eq:GammaRatio:Limit:a} and \eqref{Eq:SumEps:Limit} into \eqref{Eq:Density:UB} gives 
\begin{equation}
\limsup_{L\rightarrow\infty} \frac{\lambda}{L^{1-\frac{2}{\alpha}}} \leq \frac{2}{\pi}\l\{\frac{\E[W]}{\theta(1-\epsilon)}\r\}^{\frac{2}{\alpha}}. \label{Eq:LambdaLim:UB}
\end{equation}
The desired result follows from \eqref{Eq:LambdaLim:LB} and \eqref{Eq:LambdaLim:UB}
as well as the TC definition.

\subsection{Proof of Lemma~\ref{Lem:AsymPout:PerfCSI}} \label{App:AsymPout:PerfCSI}
For $\lambda \rightarrow 0$, we can obtain from \eqref{Eq:Pout:PerCSI:LB} that 
\begin{eqnarray}
\Pout^\ell(\lambda) &=& \sum_{k=L+1}^\infty  \frac{\l(\nu\lambda\theta^{\frac{2}{\alpha}}\r)^k}{\Gamma(k+1)}\E\l[\l(\frac{W}{\delta\sp}\r)^{-\frac{2k}{\alpha}}e^{-\nu\lambda\theta^{\frac{2}{\alpha}}\l(\frac{W}{\delta\sp}\r)^{-\frac{2}{\alpha}}}\r]\nn\\
&=&\E\l[\l(\frac{W}{\delta\sp}\r)^{-\frac{2}{\alpha}(L+1)}\r]\frac{\l(\nu\lambda\theta^{\frac{2}{\alpha}}\r)^{L+1}}{\Gamma(L+2)} + O(\lambda^{L+2}) \label{eq:pouta}\\
&=& \kappa_1\lambda^{L+1} + O(\lambda^{L+2})  \label{Eq:PoutLB:Asymp}
\end{eqnarray}
where~\eqref{eq:pouta} uses the distributions of $W$ and $\delta\sp$ in \eqref{Eq:PDF:W} and Lemma~\ref{Lem:Delta}, respectively, and $\kappa_1$ is defined in the lemma statement. 
The first inequalities in \eqref{Eq:PoutScale:a} and \eqref{Eq:PoutScale:b} follow from the last equation. 

Next, we prove the second inequalities in \eqref{Eq:PoutScale:a} and \eqref{Eq:PoutScale:b} as follows. Similar to \eqref{Eq:PoutLB:Asymp}, for $\lambda\rightarrow 0$,  $\Lambda_1$ in  \eqref{eq:1} is obtained as
\begin{eqnarray}
\Lambda_1 &=&\frac{E\l[W^{-\fa(L+1)}\r]\l(\nu\theta^{\frac{2}{\alpha}}\r)^{L+1}}{\Gamma(L+2)} \lambda^{L+1}+ O(\lambda^{L+2}). \label{Eq:Lambda:1}
\end{eqnarray}
The asymptotic expression for  $\Lambda_2$ in \eqref{eq:2} is derived as 
\begin{eqnarray}
\Lambda_2(\lambda) &=& \frac{2(\nu\lambda)^{L+1}}{\alpha\Gamma(L+1)}\int_0^\infty \int_{w\theta^{-1} + O(\lambda)}^{w\theta^{-1}} g^{-\frac{2(L+1)}{\alpha} -1}e^{-\nu\lambda g^{-\frac{2}{\alpha}}}d g f_W(w)dw\label{Eq:Lambda:2a}\\
&=& \frac{2(\nu\lambda)^{L+1}}{\alpha\Gamma(L+1)}\int_0^\infty \l[\l(w\theta^{-1}\r)^{-\frac{2(L+1)}{\alpha}-1} + O(\lambda)\r]\times O(\lambda) f_W(w)dw\nn\\
&=& O(\lambda^{L+2} ). 
 \label{Eq:Lambda:2}
\end{eqnarray}
where \eqref{Eq:Lambda:2a} uses \eqref{Eq:ExpInterfPwr}. To derive the asymptotic expression for $\Lambda_3$ in \eqref{eq:3}, it is split into two terms as $\Lambda_3 = \Lambda_{3,1} + \Lambda_{3,2}$ where 
\begin{eqnarray}
\Lambda_{3,i} = \iint\limits_{(w,g)\in\mathcal{D}_{2,i}}\min\left\{\frac{\var(I \,|\, \Ip=g)}{\left\{w\theta^{-1} - g -\E\left[I \,|\, \Ip=g\right]\right\}^2}, 1\right\} f_W(w)f\sp(g)dwdg
\label{Eq:Lambda:3b}
\end{eqnarray}
with 
\begin{eqnarray}
\mathcal{D}_{2,1}  &=& \left\{ (w, g) \mid  \frac{w \theta^{-1} }{2}\leq g + \E[I\ss\mid  \Ip = g] <  w \theta^{-1}   \right\} \nn\\
\mathcal{D}_{2,2}  &=& \left\{ (w, g) \mid  0\leq g + \E[I\ss\mid  \Ip = g] < \frac{w \theta^{-1} }{2}  \right\}. \label{Eq:D2b:Def}
\end{eqnarray} 
For $\lambda\rightarrow 0$, $\Lambda_{3,1}$ is obtained as 
\begin{eqnarray}
\Lambda_{3,1} &\leq &\iint\limits_{(w,g)\in\mathcal{D}_{2,1}} f_W(w)f\sp(g)dwdg\nn\\
&=& \frac{(\nu\lambda)^{L+1}}{\Gamma(L+1)}\int_0^\infty\int_{(w\theta^{-1})^{-\frac{2}{\alpha}} + O(\lambda)}^{\l(\frac{w\theta^{-1}}{2}\r)^{-\frac{2}{\alpha}} + O(\lambda)}g^Ldgf_W(w)dw + O(\lambda^{L+2})\nn\\
&=& \frac{\l[2^{\frac{2(L+1)}{\alpha}}-1\r]\E\l[W^{-\frac{2(L+1)}{\alpha}}\r]\l(\theta^{\frac{2}{\alpha}}\nu\r)^{L+1}}{\Gamma(L+2)} \lambda^{L+1}+ O(\lambda^{L+2}).\label{Eq:Lambda:3a}
\end{eqnarray}
Next, $\Lambda_{3,2}$ defined in \eqref{Eq:Lambda:3b} is upper bounded as 
\begin{eqnarray}
\Lambda_{3,2} &\leq& \iint\limits_{(w,g)\in \mathcal{D}_{2,2}} \frac{\var(I\ss \mid \Ip=g)}{(w\theta^{-1} - g - \E[I\ss \mid \Ip=g])^2}f\sp(g)f_W(w)dw dg\nn\\
&{\leq}& \int_0^\infty \int^{\frac{w \theta^{-1} }{2} + O(\lambda)}_0  \frac{4\var(I\ss \mid \Ip=g)}{(w\theta^{-1})^2}f\sp(g)dgf_W(w)dw.\label{Eq:Lambda3:b}
\end{eqnarray}
where \eqref{Eq:Lambda3:b} holds since $g + \E[I\ss\mid  \Ip = g] < \frac{w \theta^{-1} }{2} $ according to \eqref{Eq:D2b:Def}. 
To simplify notation, define
$
\eta =  \frac{8\theta^2 }{N(N+1)(\alpha-1)\Gamma(L+1)}. \label{Eq:Eta:Def}
$
Substituting the distribution functions in \eqref{Eq:PDF:G} and \eqref{Eq:IPwrSec:Var}  into \eqref{Eq:Lambda3:b} gives 
\begin{eqnarray}
\Lambda_{3,2} &\leq& \frac{2}{\alpha}\eta (\nu\lambda)^{L+2} \int_0^\infty w^{-2}\int_0^{\frac{w \theta^{-1} }{2} + O(\lambda)}g^{-\frac{2}{\alpha}(L+2)+1}e^{-\nu\lambda g^{-\fa}} dgf_W(w)dw\nn\\
&=& \eta(\nu\lambda)^\alpha \int_0^\infty w^{-2}\int^\infty_{\l(\frac{w\theta^{-1}}{2}\r)^{-\fa}\nu\lambda + O(\lambda^2)}g^{L-\alpha +1}e^{-g} dg f_W(w)dw. \label{Eq:Lambda3:c}
\end{eqnarray}
For $L-\alpha +1 < 0$, we obtain using \eqref{Eq:Lambda3:c} that 
\begin{eqnarray}
\Lambda_{3,2} &\leq& \eta(\nu\lambda)^\alpha \int_0^\infty w^{-2}\l[\l(\frac{w\theta^{-1}}{2}\r)^{-\fa}\nu\lambda\r]^{L-\alpha +1}f_W(w)dw + O(\lambda^{L+2})\nn\\
&=& \frac{2^{\frac{2}{\alpha}(L+1)+1}(L+1)}{N(N+1)(\alpha-1)}\times \frac{\E\l[W^{-\frac{2}{\alpha}(L+1)}\r](\theta^{\frac{2}{\alpha}}\nu)^{L+1}}{\Gamma(L+2)}\lambda^{L+1}+O(\lambda^{L+2}). \label{Eq:Lambda:3b:Case2}
\end{eqnarray}
Since $\Pout^u \leq \Lambda_1 + \Lambda_2 + \Lambda_{3,1} + \Lambda_{3,2}$, 
\begin{equation}
\limsup_{\lambda\rightarrow 0 }\frac{\Pout^u}{\lambda^{L+1}} \leq \liminf_{\lambda\rightarrow 0 }\frac{\Lambda_1}{\lambda^{L+1}} + \liminf_{\lambda\rightarrow 0 }\frac{\Lambda_2}{\lambda^{L+1}}+\liminf_{\lambda\rightarrow 0 }\frac{\Lambda_{3,1}}{\lambda^{L+1}} + \liminf_{\lambda\rightarrow 0 }\frac{\Lambda_{3,2}}{\lambda^{L+1}}. \label{Eq:Lambda:Limit}
\end{equation}
The substitution of  \eqref{Eq:Lambda:1},  \eqref{Eq:Lambda:2}, \eqref{Eq:Lambda:3a} and  \eqref{Eq:Lambda:3b:Case2}  into \eqref{Eq:Lambda:Limit} gives that 
\begin{eqnarray}
\limsup_{\lambda\rightarrow 0 }\frac{\Pout^u}{\lambda^{L+1}} &\leq&\l[1+\frac{2(L+1)}{N(N+1)(\alpha-1)}\r]\frac{2^{\frac{2}{\alpha}(L+1)}\E\l[W^{-\frac{2}{\alpha}(L+1)}\r](\theta^{\frac{2}{\alpha}}\nu)^{L+1}}{\Gamma(L+2)} \nn\\
&{\leq} & 2\times\frac{2^{\frac{2}{\alpha}(L+1)}\E\l[W^{-\frac{2}{\alpha}(L+1)}\r](\theta^{\frac{2}{\alpha}}\nu)^{L+1}}{\Gamma(L+2)}
\end{eqnarray}
where the last inequality holds since $N \geq L +1$ and $\alpha > 2$. 
The second inequality  in \eqref{Eq:PoutScale:a} follows.  
For $L-\alpha +1 \geq 0$,  $\Gamma(L-\alpha +2)$ is finite and hence we obtain  from \eqref{Eq:Lambda3:c} that 
\begin{equation}
\Lambda_{3,2} \leq \kappa_3 \lambda^\alpha+ O(\lambda^{\alpha + 1}) \label{Eq:Lambda:3b:Case1}
\end{equation}
where $\kappa_3$ is defined in the lemma statement. 
Substituting \eqref{Eq:Lambda:1},  \eqref{Eq:Lambda:2}, \eqref{Eq:Lambda:3a} and \eqref{Eq:Lambda:3b:Case1} into \eqref{Eq:Lambda:Limit}  gives the second equality in \eqref{Eq:PoutScale:b}. This completes the proof.

\subsection{Proof of Theorem~\ref{Theo:PoutBnds:ImpCSI}}\label{App:PoutBnds:ImpCSI}
{Using \eqref{Eq:SIR_b} and by definition, the outage probability for  imperfect CSI is given as 
\begin{equation}
\tPout(\tilde{\theta}) =  \E\left[ \Pr\left(\left. W\tilde{\theta}^{-1} \leq \sigma^2_{R} + I_{\Sigma}\right| \Phi, \{\bh_{T}\}\right)\right]. \label{Eq:App:Pout}
\end{equation}
Conditioned on $\Phi$ and $\{\bh_{T}\}$, the randomness of the residual interference $I_R$ in \eqref{Eq:ResIntfPwr:Pwr} depends only on the data symbols $\{x_T\}$ that follow i.i.d. $\mathcal{CN}(0, 1)$ distributions. Thus, the conditional variance $\sigma^2_{R}$ of $I_R$ can be written as 
\begin{eqnarray}
\sigma_{R}^2 &=& \frac{1}{M} \sum_{T\in\mathcal{T}} \left|\sum_{T'\in\Phi\backslash\mathcal{T}}  r_{T'}^{-\alpha/2}  \bv_0^\dagger\bG_{T'}\tilde{x}_{T',T}\right|^2\nn\\
&=&\frac{1}{M} \sum_{T\in\mathcal{T}}\left|\sum_{T'\in\Phi\backslash\mathcal{T}} I_{T'}\frac{\bv_0^\dagger\bG_{T'}\tilde{x}_{T',T}}{|\bv_0^\dagger\bG_{T'}|}\right|^2. \label{Eq:ResInterf:Var:a}
\end{eqnarray}
Let $\sim$ represent equivalence in distribution. Since   $\l\{\frac{\bv_0^\dagger\bG_{T'}\tilde{x}_{T',T}}{|\bv_0^\dagger\bG_{T'}|}\r\}$ consists of  i.i.d. $\mathcal{CN}(0,1)$ elements,   we obtain from \eqref{Eq:ResInterf:Var:a} that conditioned on  $\Phi$ and  $\{\bh_T\}$, 
\begin{eqnarray}
\sigma_{R}^2 &\sim& \frac{1}{M}  \sum_{T'\in\Phi\backslash\mathcal{T}} I_{T'} \sum_{T\in\mathcal{T}}z_{T}\nn\\
&\sim&  \frac{\zeta }{M}I_{\Sigma}\label{Eq:ResInterf:Var:b}
\end{eqnarray}
where  $\{z_T \}$ are   i.i.d. exponential random variables  with unit mean and  $\zeta$ is a chi-square random variable having  $L$ complex degrees of freedom.
 By substituting \eqref{Eq:ResInterf:Var:b}
 into \eqref{Eq:App:Pout}, we obtain that 
\begin{equation}
\tPout(\tilde{\theta})    = \Pr\left(W\tilde{\theta}^{-1} \leq I_{\Sigma} \left(1 + \frac{\zeta}{M}\right)\right).\label{Eq:Pout:App:c}
\end{equation}
Given $Z> 0$, the above expression can be expanded as  
\begin{eqnarray}
\tPout(\tilde{\theta})    &=& \Pr\left(\left. W\tilde{\theta}^{-1} \leq I_{\Sigma} \left(1 + \frac{\zeta}{M}\right)\right| \zeta \leq Z \right)\Pr(\zeta \leq  Z )+\\
&& \Pr\left(\left. W\tilde{\theta}^{-1} \leq I_{\Sigma} \left(1 + \frac{\zeta}{M}\right)\right| \zeta > Z \right)\Pr(\zeta > Z )\nn\\
&\leq&  \Pr\left(W\tilde{\theta}^{-1} \leq I_{\Sigma} \left(1 + \frac{Z}{M}\right)\right)+ \Pr(\zeta > Z ).
\label{Eq:App:f}
\end{eqnarray}
By setting $\theta = \l(1+\frac{Z}{M}\r)\tilde{\theta}$,  the inequality in \eqref{Eq:App:f} reduces to 
\begin{equation}
\tPout(\tilde{\theta}) \leq \Pout(\theta) + \Pr(\zeta > Z ). \nn
\end{equation}
It follows that 
\begin{eqnarray}
\Delta P &\leq& \Pr(\zeta > Z )\nn\\
&\leq& 1-\left(1-e^{-\omega Z}\right)^{L}\label{Eq:Alzer}\\
&\leq& Le^{-\omega Z}\label{Eq:App:j}
\end{eqnarray}
where $\omega$ is defined in the theorem statement,  \eqref{Eq:Alzer} applies Alzer's inequalities for the incomplete Gamma function  \cite{Alzer:GamFunIneq:97}, and \eqref{Eq:App:j} uses Bernoulli's inequality. 
Moreover, given $\theta = \l(1+\frac{Z}{M}\r)\tilde{\theta}$,  the rate loss is bounded  as
\begin{eqnarray}
\Delta B &=& \log_2(1+\theta) - \log_2\(1+\frac{\theta}{1+\frac{Z}{M}}\)\nn\\
&\leq& \log_2(1+\theta) - \log_2\(\frac{1+\theta}{1+\frac{Z}{M}}\)\nn\\
&=& \log_2\(1+\frac{Z}{M}\).\label{Eq:App:k}
\end{eqnarray}
From \eqref{Eq:App:j} and  \eqref{Eq:App:k}, to satisfy the constraints $\Delta P \leq \vartheta_p$ and $\Delta B \leq \vartheta_b$, it is sufficient that
\begin{eqnarray}
Le^{-\omega Z} &=& \vartheta_{p}\nn\\
\log_2\(1+\frac{Z}{M}\) &=& \vartheta_{b}. 
\end{eqnarray}
Solving the above equations gives the training sequence length in \eqref{Eq:TrainLen}. 

Finally, from  \eqref{Eq:TxCap} and \eqref{Eq:TXCap:ImpCSI}, the capacity loss defined in \eqref{Eq:CapDiff} is upper bounded as 
\begin{eqnarray}
\Delta C &=& (1-\epsilon) \lambda\log_2(1+\theta) - (1-\tPout) \lambda \log_2(1+\tilde{\theta})\nn\\
&=& \lambda\l[(1-\epsilon)  \log_2(1+\theta) - (1-\tPout)  \log_2(1+\theta) + \right.\nn\\
&&\l. (1-\tPout)  \log_2(1+\theta) -(1-\tPout) \log_2(1+\tilde{\theta})\r]\nn\\
&=& \lambda\l[\Delta P  \log_2(1+\theta)  + (1-\tPout)  \Delta B\r]\nn\\
&=& C\l[\frac{\Delta P}{1-\epsilon}   + \frac{1-\tPout}{1-\epsilon} \times\frac{ \Delta B}{\log_2(1+\theta)}\r]\nn\\
&\leq& C\l[\frac{\Delta P}{1-\epsilon}   + \frac{ \Delta B}{\log_2(1+\theta)}\r].\nn
\end{eqnarray}
The desired result in \eqref{Eq:CapLoss} follows from the above inequality, completing the proof. 

}

\subsection{Proof of Theorem~\ref{Theo:TxCap:ImpCSI}} \label{App:TxCap:ImpCSI}
{

We prove in the sequel that the training sequence length stated in the theorem achieves the TC scaling  in \eqref{Eq:CapLaw:a} with $\Delta P \overset{\epsilon}{\rightarrow} 0$ and $\Delta B \overset{\epsilon}{\rightarrow} 0$. The parallel proof concerning the other capacity scaling in \eqref{Eq:CapLaw:b}  is similar and omitted for brevity. Recall that $\lambda$ is fixed regardless of whether CSI is perfect. Thus, it follows from \eqref{Eq:CapLaw:a} that for sufficiently small $\epsilon$, 
\begin{equation}\label{Eq:CapLaw:ImpCSI}
\kappa_2^{-\frac{1}{L+1}}\leq  \frac{\tilde{C}(\epsilon)}{\log(1+\tilde{\theta})\epsilon^{\frac{1}{L+1}}}  \leq \kappa_1^{-\frac{1}{L+1}}. 
\end{equation}
As in Appendix~\ref{App:PoutBnds:ImpCSI}, we set $\theta = \l(1+\frac{Z}{M}\r)\tilde{\theta}$ with $Z > 0$ and thus \eqref{Eq:App:j} holds. Furthermore, 
we choose $Z$ and $M$ such  that $Le^{-\omega Z} = \epsilon^{1+\varrho}$ and $\frac{Z}{M} = \epsilon^{\varrho}\log\frac{1}{\epsilon}$ with $\varrho >0$, which yields  $\Delta P \overset{\epsilon}{\rightarrow} 0$ as a result of  \eqref{Eq:App:j} and has two other consequences:
\begin{eqnarray}
M &=&  \frac{1+\varrho}{\omega}\epsilon^{-\rho} + \frac{\log L}{\omega}\times\frac{\epsilon^{-\varrho}}{\log\frac{1}{\epsilon}}\label{Eq:M:Asymp}\\
&=& \frac{1+\varrho}{\omega}\epsilon^{-\rho} + o\l(\epsilon^{-\varrho}\r)\label{Eq:App:o}\\
\lim_{\epsilon \rightarrow 0} \frac{\theta}{\tilde{\theta}} &=& \lim_{\epsilon \rightarrow 0} \l(1+\frac{Z}{M}\r)\nn\\
&=&\lim_{\epsilon \rightarrow 0} \(1+\epsilon^{\varrho}\log\frac{1}{\epsilon}\) \nn \\
&=& 1.\label{Eq:App:m}
\end{eqnarray}
It follows from \eqref{Eq:App:m} that $\Delta B\overset{\epsilon}{\rightarrow} 0$ and $\lim_{\epsilon \rightarrow 0} \frac{\log_2(1+\theta)}{\log_2(1+\tilde{\theta})} = 1$. Combining the last inequality and \eqref{Eq:CapLaw:ImpCSI} gives that for sufficiently small $\epsilon$, 
\begin{equation}\label{Eq:CapLaw:ImpCSI:a}
\kappa_2^{-\frac{1}{L+1}}\leq  \frac{\tilde{C}(\epsilon)}{\log(1+\theta)\epsilon^{\frac{1}{L+1}}}  \leq \kappa_1^{-\frac{1}{L+1}} 
\end{equation}
Therefore, the TC scaling for imperfect CSI is identical to the perfect-CSI counterpart in \eqref{Eq:CapLaw:a}.  Furthermore, the scaling of $M$ in  \eqref{Eq:M:Scale} follows from \eqref{Eq:M:Asymp}. Since above results hold for an arbitrary $\varrho > 0$, 
the proof is complete. 

}

\renewcommand{\baselinestretch}{1.3}
\bibliographystyle{ieeetr}

\end{document}